\documentclass[sn-mathphys, Numbered]{sn-jnl}

\usepackage{graphicx}%
\usepackage{multirow}%
\usepackage{amsmath,amssymb,amsfonts}%
\usepackage{amsthm}%
\usepackage{mathrsfs}%
\usepackage[title]{appendix}%
\usepackage{xcolor}%
\usepackage{textcomp}%
\usepackage{manyfoot}%
\usepackage{booktabs}%
\usepackage{algorithm}%
\usepackage{algorithmicx}%
\usepackage{algpseudocode}%
\usepackage{listings}%

\usepackage[utf8]{inputenc}
\usepackage{scalerel}
\usepackage{diagbox}
\usepackage{hyperref}
\usepackage{fancyhdr}
\usepackage[all]{xy}
\usepackage{dutchcal}
\usepackage{pifont}
\usepackage{IEEEtrantools}

\usepackage[colorinlistoftodos,textwidth=25mm]{todonotes}
\usepackage{bm}
\usepackage{todonotes}
\usepackage[normalem]{ulem}

\usepackage{afterpage} 
\usepackage{lscape} 
\usepackage{pdflscape} 
\usepackage{rotating} 
\usepackage{pbox} 
\usepackage{longtable}

\makeindex

\usepackage[url,hyperrefblack,notheorems,IEEEtran]{research17}
\usepackage{nicematrix}
\usepackage{nicefrac}
\newcommand{\stkout}[1]{\ifmmode\text{\sout{\ensuremath{#1}}}\else\sout{#1}\fi}

\newcommand{\Lwt}[1]{w_{\tn{Lee}}\left(#1\right)} 

\newcommand{\notransversal}{\mathbin{\hbox{$\cap$ \kern-.45cm
\raise.3ex\hbox{$\top$}\kern-.26cm \raise.3ex\hbox{$/$}}}}
\newcommand{\oplusq}[1]{\stackrel{\scaleto{#1}{4pt}}{\oplus}} 
\newcommand{\odotq}[1]{\stackrel{\scaleto{#1}{4pt}}{\odot}} 
\newcommand{\gray}{\mathcal{g}} 
 
\newcommand{\cmark}{\ding{51}}%
\newcommand{\xmark}{\ding{55}}

\newtheorem{theorem}{\mytheoremname}
\newtheorem{lemma}[theorem]{\mylemmaname}
\newtheorem{corollary}[theorem]{\mycorollaryname}

\newtheorem{proposition}[theorem]{\mypropositionname}

\newtheorem{definition}[theorem]{\mydefinitionname}
\newtheorem{remark}[theorem]{\myremarkname}
\newtheorem{example}[theorem]{\myexamplename}

\vspace{25cm}

\begin{document}

\title[Linearity of $\Integers_{2^L}$-Linear Codes via Schur Product]{Linearity of $\Integers_{2^L}$-Linear Codes via Schur Product}


\author[1]{\fnm{Gustavo} \sur{T. Bastos}}\email{gtbastos@ufsj.edu.br}
\equalcont{These authors contributed equally to this work.}

\author*[2]{\fnm{Maiara F.} \sur{Bollauf}}\email{maiara.bollauf@ut.ee}

\author[3]{\fnm{Agnaldo J.} \sur{Ferrari}}\email{agnaldo.ferrari@unesp.br}
\equalcont{These authors contributed equally to this work.}

\author[4]{\fnm{\O yvind} \sur{Ytrehus}}\email{oyvind.ytrehus@uib.no}
\equalcont{These authors contributed equally to this work.}

\affil[1]{\orgdiv{Department of Mathematics and Statistics}, \orgname{Federal University of São João del-Rei}, \state{MG}, \country{Brazil}}

\affil[2]{\orgdiv{Institute of Computer Science}, \orgname{University of Tartu}, \city{Tartu}, \country{Estonia}}

\affil[3]{\orgdiv{Department of Mathematics, School of Sciences}, \orgname{S\~ao Paulo State University -- Campus Bauru}, \state{SP}, \country{Brazil}}

\affil[4]{\orgdiv{Department of Informatics}, \orgname{University of Bergen}, \city{Bergen}, \country{Norway}}


\abstract{We propose an innovative approach to investigating the linearity of $\Integers_{2^L}$-linear codes derived from $\mathbb{Z}_{2^L}$-additive codes using the generalized Gray map. To achieve this, we define two related binary codes: the associated and the decomposition codes. By considering 
the Schur product between codewords, we can determine the linearity of the respective $\mathbb{Z}_{2^L}$-linear code. As a result, we establish a connection between the linearity of the $\Integers_{2^L}$-linear codes with the linearity of the decomposition code for $\Integers_4$ and $\Integers_8$-additive codes. Furthermore, we construct $\mathbb{Z}_{2^L}$-additive codes from \emph{nested} binary codes, resulting in linear $\Integers_{2^L}$-linear codes. This construction involves multiple layers of binary codes, where a code in one layer is the \emph{square} of the code in the previous layer. 
We also present a sufficient 
condition that allows checking nonlinearity of the $\mathbb{Z}_{2^L}$-linear codes by simple binary operations in their respective associated codes. Finally, we employ our arguments to verify the linearity of well-known  $\mathbb{Z}_{2^L}$-linear code constructions, including the Hadamard, simplex, and MacDonald codes.}


\maketitle

{\def\thefootnote{}\footnotetext{This paper was partially published in the proceedings of the IEEE International Symposium on Information Theory (ISIT'24), Athens, Greece, 2024~\cite{Bastos-Bollauf-Ferrari-Ytrehus-24}.}}

\section{Introduction}

\IEEEPARstart{T}{he} Gray code was patented by Frank Gray~\cite{Gray53} in 1953. 
The code allows traversing the set of binary representations of integers between $0$ and $2^L-1$ in such a way that 
two consecutive representations differ in a single 
bit position. Hammons \emph{et.al.}~\cite{Hammons_etal_94} recognized that the Gray code for $L=2$ possesses an interesting extra property: consider the ring of integers modulo 4, $\Integers_4=\{0,1,2,3\}$, then the standard Gray map $\gray: \Integers_4 \to \Integers_2^2$ is a bijection defined as
\begin{IEEEeqnarray*}{c}
0 \mapsto 00,~ 1 \mapsto 01,~ 2 \mapsto 11,~ 3 \mapsto 10.  
\end{IEEEeqnarray*}
This map can be naturally extended to a vector in $\Integers_4^n$. Indeed $\gray$ is a distance-preserving map, i.e., the Lee distance between two vectors $\vect{u}, \vect{v} \in \Integers_4^n$ is the same as the Hamming distance of their respective Gray maps $\gray(\vect{u}), \gray(\vect{v}) \in \Integers_2^{2n}$ (for more details on weight and distance properties, see~\cite{Hammons_etal_94}).

Several authors generalized the Gray map. Carlet~\cite{carlet} extended it to $\mathbb{Z}_{2^L}$, for $L \geq 2$. Greferath and Schmidt~\cite{GreferathSchmidt_99} studied the Gray isometry in the context of finite chain rings. More recently, Heng and Yue~\cite{HengYue_15} generalized Carlet's result to $\Integers_{p^L}$ for $p$ prime, and Gupta~\cite{Gupta22} introduced yet a different generalization using the theory of modular rings. In this paper we consider Carlet's generalized Gray map $\gray$ applied to a $\Integers_{2^L}$-additive code\footnote{A  $\mathbb{Z}_{2^L}$-additive code is simply an additive subgroup of $\mathbb{Z}_{2^L}^n$.} $\code{C}$, where the image $\gray(\code{C})$ is referred to as \textit{$\Integers_{2^L}$-linear code}. 
Observe that $\Integers_{2^L}$-linear codes are \emph{binary codes that may or may not be linear.}

A comprehensive list of outstanding linear and nonlinear $\Integers_{2^L}$-linear codes obtained via the Gray map can be found in~\cite{AlbertFaist}. Furthermore, the theoretical problem of determining whether a $\Integers_{2^L}$-linear code is linear or nonlinear has been addressed in several influential papers~\cite{carlet, nonlinearity, tapiavega, Eyvazi22} in the area of algebraic coding theory. In 
Tapia-Recillas and Vega's paper~\cite{tapiavega}, the authors provide a criterion to verify the linearity of a $\mathbb{Z}_{2^L}$-linear code $\code{C}$ based on the structure of the $\Integers_{2^L}$-additive code.
Notably, for $\mathbb{Z}_{8}$-linear codes, a significant reduction in the number of pairs of codewords to be analyzed was presented in~\cite{Eyvazi22}. The main result in~\cite{Eyvazi22} states that the linearity of a $\Integers_{8}$-linear code can be verified via a subset obtained from some rows of its generator matrix $\mat{G}_{\code{C}}$. In addition, the authors analyzed conditions on generator polynomials of some cyclic codes so that their respective $\Integers_{8}$-linear codes are also linear. 

The linearity of the Gray map of other families of $\Integers_{2^L}$-additive codes, including  $\mathbb{Z}_{2^L}$-additive simplex codes of types $\alpha$ and $\beta$ as well as $\Integers_{2^L}$-additive Hadamard codes, was extensively explored in the literature~\cite{nonlinearity, CordobaVelaVillanueva19, FernandezCVelaVillanueva_20}.
A partial classification of $\Integers_{2^L}$-linear codes from $\mathbb{Z}_{2^L}$-additive Hadamard codes based on the study of the kernel of a binary
 code, i.e., $K(\code{C}):=\left\{\vect{x} \in \mathbb{F}_2^n 
 : \vect{x}+ \code{C} =  \code{C} \right\}$ was carried out in~\cite{CordobaVelaVillanueva19,FernandezCVelaVillanueva_20}.
 
 Conditions on the parameters of $\mathbb{Z}_{2^L}$-additive simplex codes were presented to state the linearity of the respective $\Integers_{2^L}$-linear codes~\cite{nonlinearity}. The generalized Hamming weights of such $\Integers_{2^L}$-linear codes were discussed in~\cite{GuptaBhandariLal_05}.
 
 An extension to another classical family, the MacDonald codes over $\mathbb{Z}_{2^L}$, is also addressed in~\cite{nonlinearity}. The generator matrices of $\mathbb{Z}_{2^L}$-additive MacDonald codes of types $\alpha$ and $\beta$ are derived from the generator matrices of $\mathbb{Z}_{2^L}$-additive simplex codes of types $\alpha$ and $\beta$, respectively. Thus, the linearity of $\Integers_{2^L}$-linear codes from this family can be analyzed through methods similar to those used for $\mathbb{Z}_{2^L}$-additive Hadamard codes. 

 In this paper, we further explore the linearity properties of binary codes obtained from the generalized Gray map. 
 Our main contributions are summarized as follows:
 \begin{enumerate}
    \item[i)] Given a $\Integers_{2^L}$-additive code $\code{C}$, we introduce (Definitions~
 \ref{def:associated_codes} and \ref{def:concatenated_code}) 
 two auxiliary binary codes induced by $\code{C}$: the \emph{associated} and \emph{decomposition} codes. 
    \item[ii)]  
We introduce conditions (Theorem~\ref{thm:iff_condition_concatenated}) to guarantee 
the linearity of $\Integers_{2^L}$-linear codes based on their binary decomposition together with the \emph{Schur} (coordinate-wise or Hadamard) product of the respective binary codewords. 
As an application, we use our conditions to revisit the linearity 
of $\Integers_{2^L}$-linear codes obtained from $\mathbb{Z}_{2^L}$-additive Hadamard, simplex, and MacDonald families. 
    \item[iii)] For $\Integers_4$ and $\Integers_8$-additive codes, we explore the relationship between the linearity of their respective $\Integers_{2^L}$-linear codes ($L=2,3$) and the linearity of the decomposition code. We also provide a sufficient condition (Lemma~\ref{lem:simple-condition})  for the linearity of $\Integers_4$-linear codes based solely on its generator matrix.
    \item[iv)] 
    We state a sufficient condition (Corollary~\ref{coro:princ})for a 
    $\Integers_{2^L}$-linear code to be \emph{non}linear that can be checked simply  by  performing binary 
    multiplications on the codewords. 
    \item[v)] As noted above, for $L=2$ the Lee distance distribution of an $\Integers_{2^L}$-linear code is equivalent to the Hamming distance distribution of its Gray map  code, but this is not true  in general for $L>2$. We discuss relations between Lee and Hamming weights also for the generalized Gray map (Proposition~\ref{prop:lee-hammimng-weight} and Corollary~\ref{cor:distpreserved}).
    \item[vi)] 
We propose the \textit{nested construction} of 
$\mathbb{Z}_{2^L}$-additive codes. The corresponding 
$\Integers_{2^L}$-linear codes are linear (Theorem~\ref{thm:concatenated-cartesian-product}). 

    \item[vii)] The linearity (and nonlinearity) of $\Integers_{2^L}$-linear codes obtained from simplex, Hadamard, and MacDonald codes is reassessed via our conditions. 
\end{enumerate}

The paper is organized as follows: Sec.~\ref{sec:shur_construction} introduces fundamental notions of linear codes, Schur product, and the generalized Gray map. Sec.~\ref{sec:linearity_z2l} presents conditions to verify the linearity of $\Integers_{2^L}$-linear codes, specialized results for $\Integers_4$ and $\Integers_8$-linear codes, and distance properties. The nested construction of linear $\Integers_{2^L}$-linear codes 
is described in Sec.~\ref{sec:linear-gray-codes}, as well as applications of this construction to Reed-Muller and cyclic codes. 
Sec.~\ref{sec:families_gray_codes} revisits the linearity of known families of $\Integers_{2^L}$-linear codes via our techniques. Finally, Sec.~\ref{sec:conclusion} concludes the paper and indicates perspectives of future work.


\section{Preliminaries}
\label{sec:shur_construction}

\subsection{Notation}
\label{sec:notation}

The ring of integers modulo $q$ is $\mathbb{Z}_q=\{0,1,\dots, q-1\}$. Vectors are boldfaced, e.g., $\vect{x}$. The symbol $+$ represents the element-wise addition over $\Reals$, $\oplusq{q}$ and $\odotq{q}$ represent the addition and multiplication over $\Integers_q$, respectively. For $a,b\in\mathbb{Z}$, $a\leq b$, we write $[a:b] := \{a,a+1,\ldots,b\}$. Matrices are represented by capital sans serif letters, e.g. $\mat{X}$. The $n \times n$ identity matrix is denoted by $\mat{I}_n$. The span of a set of vectors $\vect{v}_1,\dots,\vect{v}_n$ is denoted by $\langle \vect{v}_1,\dots,\vect{v}_n  \rangle$. We use $\Integers_2$ with its field structure, not just the ring. 


Please be advised that throughout the paper, we will refer to different mappings from $\mathbb{Z}_{2^L}$ (and, by extension, from  $\mathbb{Z}_{2^L}^n$) to other structures. To avoid confusion, we included a sketch of these mappings and where they are used in Figure~\ref{fig:mappings}.

\tikzset{every picture/.style={line width=0.75pt}} 

\begin{figure}[t]
\centering

\scalebox{0.8}{    
\begin{tikzpicture}[x=0.75pt,y=0.75pt,yscale=-1,xscale=1]

\draw  [color={rgb, 255:red, 155; green, 155; blue, 155 }  ,draw opacity=1 ][line width=1.5]  (302,58) .. controls (302,51.37) and (307.37,46) .. (314,46) -- (345,46) .. controls (351.63,46) and (357,51.37) .. (357,58) -- (357,68.5) .. controls (357,75.13) and (351.63,80.5) .. (345,80.5) -- (314,80.5) .. controls (307.37,80.5) and (302,75.13) .. (302,68.5) -- cycle ;
\draw  [color={rgb, 255:red, 155; green, 155; blue, 155 }  ,draw opacity=1 ][line width=1.5]  (85,151) .. controls (85,144.37) and (90.37,139) .. (97,139) -- (128,139) .. controls (134.63,139) and (140,144.37) .. (140,151) -- (140,161.5) .. controls (140,168.13) and (134.63,173.5) .. (128,173.5) -- (97,173.5) .. controls (90.37,173.5) and (85,168.13) .. (85,161.5) -- cycle ;
\draw  [color={rgb, 255:red, 155; green, 155; blue, 155 }  ,draw opacity=1 ][line width=1.5]  (304,153) .. controls (304,146.37) and (309.37,141) .. (316,141) -- (347,141) .. controls (353.63,141) and (359,146.37) .. (359,153) -- (359,163.5) .. controls (359,170.13) and (353.63,175.5) .. (347,175.5) -- (316,175.5) .. controls (309.37,175.5) and (304,170.13) .. (304,163.5) -- cycle ;
\draw  [color={rgb, 255:red, 155; green, 155; blue, 155 }  ,draw opacity=1 ][line width=1.5]  (500,154) .. controls (500,147.37) and (505.37,142) .. (512,142) -- (543,142) .. controls (549.63,142) and (555,147.37) .. (555,154) -- (555,164.5) .. controls (555,171.13) and (549.63,176.5) .. (543,176.5) -- (512,176.5) .. controls (505.37,176.5) and (500,171.13) .. (500,164.5) -- cycle ;
\draw [color={rgb, 255:red, 137; green, 218; blue, 53 }  ,draw opacity=1 ][fill={rgb, 255:red, 126; green, 211; blue, 33 }  ,fill opacity=1 ][line width=1.5]    (299,75.5) -- (123.85,132.08) ;
\draw [shift={(121,133)}, rotate = 342.1] [color={rgb, 255:red, 137; green, 218; blue, 53 }  ,draw opacity=1 ][line width=1.5]    (14.21,-4.28) .. controls (9.04,-1.82) and (4.3,-0.39) .. (0,0) .. controls (4.3,0.39) and (9.04,1.82) .. (14.21,4.28)   ;
\draw [color={rgb, 255:red, 208; green, 2; blue, 27 }  ,draw opacity=1 ][fill={rgb, 255:red, 126; green, 211; blue, 33 }  ,fill opacity=1 ][line width=1.5]    (330,82.5) -- (329.69,132.33) ;
\draw [shift={(329.67,135.33)}, rotate = 270.36] [color={rgb, 255:red, 208; green, 2; blue, 27 }  ,draw opacity=1 ][line width=1.5]    (14.21,-4.28) .. controls (9.04,-1.82) and (4.3,-0.39) .. (0,0) .. controls (4.3,0.39) and (9.04,1.82) .. (14.21,4.28)   ;
\draw [color={rgb, 255:red, 74; green, 144; blue, 226 }  ,draw opacity=1 ][fill={rgb, 255:red, 126; green, 211; blue, 33 }  ,fill opacity=1 ][line width=1.5]    (357,80.5) -- (511.17,135) ;
\draw [shift={(514,136)}, rotate = 199.47] [color={rgb, 255:red, 74; green, 144; blue, 226 }  ,draw opacity=1 ][line width=1.5]    (14.21,-4.28) .. controls (9.04,-1.82) and (4.3,-0.39) .. (0,0) .. controls (4.3,0.39) and (9.04,1.82) .. (14.21,4.28)   ;

\draw (315.33,53.73) node [anchor=north west][inner sep=0.75pt]    {$\mathbb{Z}_{2^{L}}$};
\draw (101,145.4) node [anchor=north west][inner sep=0.75pt]    {$\mathbb{Z}_{2}^{L}$};
\draw (313,146.4) node [anchor=north west][inner sep=0.75pt]    {$\mathbb{Z}_{2}^{2^{L-1}}$};
\draw (514.67,148.4) node [anchor=north west][inner sep=0.75pt]    {$\mathbb{Z}_{2}^{L}$};
\draw (82,181) node [anchor=north west][inner sep=0.75pt]   [align=left] {{\small Gray code}};
\draw (273,185) node [anchor=north west][inner sep=0.75pt]   [align=left] {\begin{minipage}[lt]{77.22pt}\setlength\topsep{0pt}
\begin{center}
{\small  Generalized Gray}\\{\small map~\cite{carlet}, Def.~\ref{def:generalized_gray}}
\end{center}

\end{minipage}};
\draw (468,185) node [anchor=north west][inner sep=0.75pt]   [align=left] {\begin{minipage}[lt]{98.47pt}\setlength\topsep{0pt}
\begin{center}
{\small Natural map from}\\{\small binary decomposition, Sec.~\ref{sec:linearity_z2l}}
\end{center}

\end{minipage}};
\end{tikzpicture}}
    
\caption{Mappings from $\Integers_{2^L}$ (and, by extension, from $\Integers_{2^L}^n$) used in this paper. The Carlet's generalized Gray map coincides with a Gray code only for $L=2$.}
\label{fig:mappings}
\vspace{-3ex}
\end{figure}
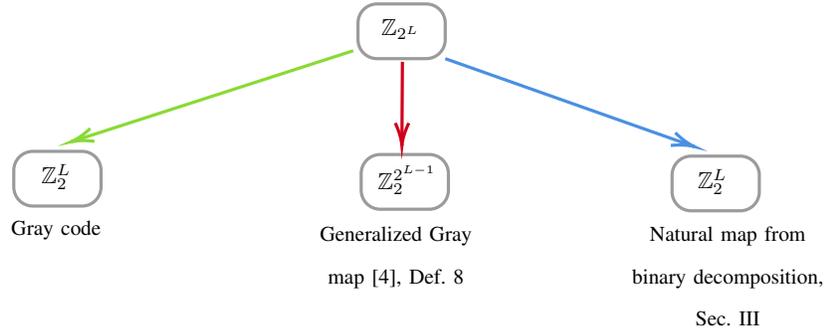

\subsection{Codes}
\label{sec:codes}

    We introduce relevant definitions of linear codes over $\Integers_{q}$, here denoted as $\Integers_q$-additive codes, and general properties of the Schur product.


\begin{definition}[$\mathbb{Z}_{q}$-additive code]\label{def:zq-additive}
A code $\code{C} \subseteq \mathbb{Z}_{q}^n$ is said to be $\mathbb{Z}_{q}$-additive if it is an additive subgroup of $\mathbb{Z}_{q}^n$. For $q=2$, a $\Integers_2$-additive code is simply a binary linear code.
\end{definition}

The code parameters $[n,M]$ or $[n,M,d_{\textnormal{Lee}}]$ indicate a $\Integers_q$-additive code $\code{C}$ of length $n$, $M$ codewords, and minimum Lee distance 
\begin{IEEEeqnarray*}{c}
d_{\textnormal{Lee}}=d_{\textnormal{Lee}}(\code{C}) \eqdef \displaystyle\min_{\vect{x},\vect{y}\in\code{C}}\Lwt{\vect{x}-\vect{y}} = \displaystyle\min_{\vect{x},\vect{y}\in\code{C}} \sum_{i=1}^n \min\{(x_i-y_i),q-(x_i-y_i)\}.   
\end{IEEEeqnarray*}

When referring to binary codes, the Hamming weight of a codeword $\vect{c} \in \code{C} \subseteq \Integers_2^n$ is denoted by $w_{\textnormal{H}}(\vect{c})$ and the minimum Hamming distance of a code $\code{C} \subseteq \Integers_2^n$ is $d_{\textnormal{H}}(\code{C})$. 

\begin{definition} (\textit{Schur product})\label{def:schur-product} For $\vect{x}=(x_{1}, \dots, x_{n}), \vect{y}=(y_{1}, \dots, y_{n}) \in \Integers_{2}^{n},$ we define the Schur product $\vect{x} \circ \vect{y} = (x_{1}\odotq{2} y_{1}, \dots, x_{n}\odotq{2} y_{n}) = (x_{1}y_{1}, \dots, x_{n}y_{n}) \in \Integers_{2}^n$.
\end{definition}

The subspace generated by all the Schur products of vectors of a code $\code{C}$ is denoted as the square of $\code{C}$. Square codes have been extensively studied in~\cite{CascudoCramerMirandolaZemor_15, Cascudo}.

\begin{definition}(\textit{Square of a code}) \label{def:square-code} The square of a code $\code{C} \subseteq \Integers_2^n$ is $\code{C}^{ 2} =\langle \left\{\vect{c} \circ \tilde{\vect{c}}: \vect{c}, \tilde{\vect{c}} \in \code{C} \right\} \rangle$. 
\end{definition}      
In our context, we also consider simply the Schur product of a code:  $\code{C} \circ \code{C} =\{ \vect{c} \circ \tilde{\vect{c}} : \vect{c}, \tilde{\vect{c}} \in \code{C} \}$.
Notice that $\code{C} \subseteq \code{C}\circ \code{C} \subseteq \code{C}^2$, but while $\code{C}^2$ is a linear code, $\code{C} \circ \code{C}$ is nonlinear in general.

\begin{definition}[Codes closed under Schur product] We say that the (chain of) binary linear codes $\code{C}_1, \dots, \code{C}_L \subseteq \Integers_2^n$ are (is) closed under Schur product if $\code{C}_{i} \circ \code{C}_i \subseteq \code{C}_{i+1}$, for all $i \in [1:L-1]$.
\end{definition}

\subsection{Associated Codes and Generalized Gray Map}
\label{sec:gray-map}

We start recalling the natural binary decomposition of an element in $\Integers_{2^L}$.

\begin{definition}[Binary decomposition] \label{def:bindecomp}
Let $L \in \Naturals$ and $v \in \Integers_{2^L}$. The binary decomposition of $v$ is $\displaystyle{v = \sum_{i=1}^L 2^{i-1}v_{i} \in \mathbb{Z}_{2^L}}$ with $ v_{i} \in \Integers_2$, for $i\in[1:L]$.
\end{definition}

Let $\code{C}$ be a $\mathbb{Z}_{2^L}$-additive code and $\vect{c}=(c_1 , c_2 , ..., c_n) \in \code{C}$. By performing the binary decomposition of each coordinate of $\vect{c}$, we can associate $\code{C} \subseteq \Integers_{2^L}^n$ to $L$ codes $\code{C}_1, \dots, \code{C}_L$ over $\Integers_2$, according to 
\begin{IEEEeqnarray}{rCl} 
\label{latticetoring}
\vect{c}&=&(c_1 , c_2 , ..., c_n) = \left(\sum_{i=1} ^L 2^{i-1}u_{i,1} , \sum_{i=1} ^L  2^{i-1}u_{i,2} , \ldots, \sum_{i=1} ^L  2^{i-1}u_{i,n} \right)\nonumber \\
& = & \left(u_{1,1},u_{1,2},...,u_{1,n} \right) + 2 \left(u_{2,1},u_{2,2},...,u_{2,n} \right)+\ldots+ 2^{L-1} \left(u_{L,1},u_{L,2},...,u_{L,n} \right) \nonumber\\
&=& \vect{u}_1 + 2 \vect{u}_2 + \dots + 2^{L-1}\vect{u}_L = \sum_{i=1}^L 2^{i-1}\vect{u}_i,  
\end{IEEEeqnarray} 
where $\vect{u}_i = \left(u_{i,1},u_{i,2},...,u_{i,n} \right) \in \Integers_2^n$ for $i \in [1: L]$. 

\begin{definition}[Associated codes]
\label{def:associated_codes} 
Consider a code $\code{C} \subseteq \mathbb{Z}_{2^L}^n$. For all $\vect{c} = \sum_{i=1}^L 2^{i-1}\vect{u}_i \in \code{C}$, we fix an $i \in [1: L]$, and place all the respective vectors $\vect{u}_i$ in a set $\code{C}_i$. The resulting codes $\code{C}_1, \dots, \code{C}_L \subseteq \Integers_2^n$ are called associated codes with $\code{C}$.
\end{definition}

If $\code{C} \subseteq \Integers_{2^L}^n$ is a $\Integers_{2^L}$-additive code, then its respective associated codes $\code{C}_1, \dots, \code{C}_L$ are binary linear codes~\cite[Lemma~1]{maiara}.

\begin{example}
\label{ex:tela-preta-gustavo}
Let $\code{C}=\left\langle(0,1,2,3,4,5,6,7)\right\rangle \subset \mathbb{Z}_8 ^8$ be a $\mathbb{Z}_8$-additive code. This code has three associated codes. 
For example, for $\vect{c}=(0,2,4,6,0,2,4,6)\in\code{C}$, we have $\vect{u}_{1} = (0,0,0,0,0,0,0,0)$, $\vect{u}_{2} = (0,1,0,1,0,1,0,1)$, and $\vect{u}_{3}=(0,0,1,1,0,0,1,1)$. If we proceed analogously for all codewords in $\code{C}$ and place all vectors $\vect{u}_i$ in a code $\code{C}_i$, $i=1,2,3$, we get the following associated codes
\begin{IEEEeqnarray*}{rCl}
\code{C}_1&=&\left\{(0,0,0,0,0,0,0,0),(0,1,0,1,0,1,0,1)\right\},\\
\code{C}_2&=&\left\{(0,0,0,0,0,0,0,0),(0,1,0,1,0,1,0,1),(0,0,1,1,0,0,1,1) ,(0,1,1,0,0,1,1,0)\right\},\\
\code{C}_3&=&\left\{(0,0,0,0,0,0,0,0),(0,1,0,1,0,1,0,1),(0,0,1,1,0,0,1,1) ,(0,1,1,0,0,1,1,0),\right.\\
&&\left.(0,0,0,0,1,1,1,1), (0,0,1,0,1,1,0,1),(0,1,0,1,1,0,1,0),(0,1,1,1,1,0,0,0) \right\},
\end{IEEEeqnarray*}
where, in particular, $\code{C}_1 \subset\code{C}_2 \subset \code{C}_3$.
\hfill\exampleend
\end{example}

We present next the definition of generalized Gray map~\cite{carlet}.
\begin{definition}[Generalized Gray map]
\label{def:generalized_gray}
Let $L \in \Naturals$ and $\displaystyle{v = \sum_{i=1}^L 2^{i-1}v_{i} \in \mathbb{Z}_{2^L}}$ with $ v_{i} \in \Integers_2$. The image of $v$ by the generalized Gray map is the Boolean function on $\Integers_2 ^{2^{L-1}}$
\begin{IEEEeqnarray*}{rCl}
    \gray: & ~ \mathbb{Z}_{2^L} & \rightarrow  \Integers_2 ^{2^{L-1}}  \\
       & v  & \mapsto  \gray(v) =  \underbrace{\left(v_L ,v_L ,...,v_L\right)}_{2^{L-1}~\textnormal{times}} \oplusq{2} \left(v_1, 
 v_2, \ldots, v_{L-1}\right)\mat{Y},
\end{IEEEeqnarray*}
where $\mat{Y} \in \Integers_2^{(L-1) \times 2^{L-1}}$ is a matrix whose columns are all elements of $\Integers_2^{L-1}$.
\end{definition}

The Gray map can be naturally extended into a mapping from $\Integers_{2^L}^n$ to $\Integers_2^{(2^{L-1})n}$. From now on, without loss of generality, when referring to the generalized Gray map, we mean its coordinate-wise extension, i.e., for $\vect{x}=(x_1,\dots,x_n) \in \Integers_{2^L}^n$, $\gray(\vect{x}) =(\gray(x_1), \dots, \gray(x_n)) \in \Integers_2^{(2^{L-1})n}$, where $\gray(x_i)$ follows Definition~\ref{def:generalized_gray} for all $i\in [1: n]$.

\begin{example}\label{ex:gray-z4} The standard Gray map assumes $L=2$ and $\mat{Y} = (0\,\,1) \in \Integers_2^{1 \times 2}$. Then,
\begin{IEEEeqnarray*}{rCl}
  v = 0 \mapsto \gray(v) = (0,0),\quad v=1 \mapsto \gray(v)=(0,1), \\
  v=2 \mapsto \gray(v)=(1,1),\quad v=3 \mapsto \gray(v)=(1,0). 
\end{IEEEeqnarray*}
\\[-12mm]\hfill\exampleend
\end{example}

\begin{example}\label{ex:gray_z8} Consider $L=3$ and $\mat{Y}=\begin{pNiceMatrix} 0 & 0 & 1 & 1\\
       0 & 1 &0 &1
       \end{pNiceMatrix}$. For this vector ordering, we have
\begin{IEEEeqnarray*}{rCl}
  \gray(0) = (0,0,0,0),\quad \gray(1)=(0,0,1,1), \quad \gray(2)=(0,1,0,1),\quad \gray(3)=(0,1,1,0) \\ \gray(4)=(1,1,1,1) , \quad \gray(5)=(1,1,0,0), \quad \gray(6)=(1,0,1,0), \quad \gray(7)=(1,0,0,1). 
\end{IEEEeqnarray*}
\\[-12.5mm]\hfill\exampleend
\end{example}

\begin{remark} Notice that in~\cite[p.~1543]{carlet}, the Gray map for $\Integers_8$ is slightly different than the one presented in Example~\ref{ex:gray_z8}, even though both fundamentally share the same properties. This is because the placement order of the vectors in the matrix $\mat{Y}$ matters. For instance, the matrix $\mat{Y}$ considered in~\cite[p.~1543]{carlet} is a simple permutation of the rows of the one presented in our Example~\ref{ex:gray_z8}.
\end{remark}

In terms of distance-preserving properties, for $L=2$, i.e., for $\Integers_4$-additive codes, it holds that the Hamming distance between $\gray(\vect{c})$ and $\gray(\vect{d})$ is equal to the Lee weight of $\vect{c}-\vect{d}$, for any quaternary codewords $\vect{c}$ and $\vect{d}$. For $L \geq 3$, the distance is preserved through a slightly different notion. If $\code{C}$ is a $\Integers_{2^L}$-additive code and $\vect{c}, \vect{d} \in \code{C}$, then the Hamming distance between $\gray(\vect{c})$ and $\gray(\vect{d})$ is equal to the Hamming weight of $\gray(\vect{c}-\vect{d})$ (see~\cite[Prop.~1]{carlet}). This justifies that for $L=2$ (Example~\ref{ex:gray-z4}), the Gray code and Gray map are directly connected, while for $L \geq 3$ it is not the case. Further connections between the Lee distance of a code and the Hamming distance of its respective generalized Gray map are to be discussed in Sec.~\ref{sec:distance-properties}.

\begin{definition}[$\Integers_{2^L}$-linear code] Let $\code{C}$ be a $\Integers_{2^L}$-additive code. Then $\gray(\code{C})=\{\gray(\vect{c}): \forall~\vect{c} \in \code{C}\}$ is called $\Integers_{2^L}$-linear code. 
\end{definition}


\begin{example}
\label{ex:nordstrom-robinson}
The Nordstrom-Robinson code~\cite{NordstromRobinson_67} is the $\Integers_{4}$-linear code obtained from the $\Integers_4$-additive octacode $\code{O}$~\cite{ForneySloaneTrott92_1}, generated by $\mat{G}_{\code{O}} =(\mat{I}_{4}\,\,\,\mat{M}_{4})$, where 
  \begin{IEEEeqnarray*}{c}
    \mat{M}_{4} =
    \begin{pNiceMatrix}
      3 & 1 & 2 & 1 \\
      1 & 2 & 3 & 1 \\
      3 & 3 & 3 & 2 \\
      2 & 3 & 1 & 1
    \end{pNiceMatrix}.
  \end{IEEEeqnarray*}
The Nordstrom-Robinson code $\textnormal{NR} = \gray(\code{O}) \subseteq \Integers_2^{16}$ is a nonlinear $\Integers_{4}$-linear code.
\hfill\exampleend
\end{example}

\section{Linearity of $\Integers_{2^L}$-linear Codes}
\label{sec:linearity_z2l}

This section presents our main findings on the linearity of $\Integers_{2^L}$-linear codes. We introduce decomposition codes and translate a significant result in the literature regarding the linearity of $\Integers_{2^L}$-linear codes. Additionally, we provide specific results for $\Integers_4$ and $\Integers_8$, discuss a sufficient condition that simplifies the verification of the nonlinearity, and finally address distance properties of a code $\code{C}$ and its respective $\Integers_{2^L}$-linear code.

\subsection{Linearity Based on the Decomposition Code}
\label{subsec:linearity-decomposition}

Given $\displaystyle
{v=\sum_{i=1}^L 2^{i-1}v_{i}}, \displaystyle{w=\sum_{i=1}^L 2^{i-1}w_{i}} \in \mathbb{Z}_{2^L}$, the operations $\oplus$ and $\odot$ in $\mathbb{Z}_{2^L}$ are
\begin{IEEEeqnarray}{rCl}\label{operacoes}
v \oplus w & = &\sum_{i=1}^L 2^{i-1} \left(v_i \oplusq{2} w_{i} \right), ~v \odot w = \sum_{i=1}^L 2^{i-1} \left(v_i \odotq{2} w_{i}\right).
\end{IEEEeqnarray}

According to~\cite[Prop.~2.1]{tapiavega}, the sum between elements of $\mathbb{Z}_{2^L}$ is related to~\eqref{operacoes} as 
\begin{IEEEeqnarray}{c}
\label{relacaooperacoes}
v  \oplusq{2^L} w= (v \oplus w) \oplusq{2^L} 2(v \odot w),
\end{IEEEeqnarray} 
which can be coordinate-wisely extended to vectors $\vect{v}, \vect{w} \in \Integers_{2^L}^n$.


Theorem~\ref{tapiavega} below provides a condition that a $\mathbb{Z}_{2^L}$-additive code $\code{C}$ must satisfy for its $\Integers_{2^L}$-linear code $\gray(\code{C})$ to be linear. Such a condition cancels out the natural \emph{carry} involved in the real sum of elements of $\mathbb{Z}_{2^L}$ (see~\eqref{relacaooperacoes}).

\begin{theorem}\cite[Th.~4.13]{tapiavega}\label{tapiavega}
Let $\code{C}$ be a $\Integers_{2^L}$-additive code. Then $\gray(\code{C})$ is a linear code if and only if $2(\vect{c} \odot \vect{d})= \sum_{i=1}^L 2^{i}(\vect{u}_i \odotq{2} \vect{v}_i) \in \code{C} $ for all $\vect{c} = \sum_{i=1}^L 2^{i-1}\vect{u}_i, \vect{d}=\sum_{i=1}^L 2^{i-1}\vect{v}_i \in \code{C}$.
\end{theorem}

We present next an alternative approach to Theorem~\ref{tapiavega} by expressing the $\Integers_{2^L}$-additive code via their binary decomposition and the relation with operations in~\eqref{operacoes}. 

A code $\code{C} \subseteq \mathbb{Z}_{2^L} ^n$ induces a binary decomposition code defined as follows.

\begin{definition}[Decomposition code]\label{def:concatenated_code} Let  $\code{C} \subseteq \mathbb{Z}_{2^L} ^n$ be a $\mathbb{Z}_{2^L}$-additive code and $\vect{c} \in \code{C}$. Recalling~\eqref{latticetoring}, we write
\begin{IEEEeqnarray*}{c}
\vect{c}=\left(\sum_{i=1}^{L} 2^{i-1} u_{i,1} , \sum_{i=1}^{L} 2^{i-1}u_{i,2},\ldots, \sum_{i=1}^{L} 2^{i-1}u_{i,n}  \right) = \sum_{i=1}^L 2^{i-1}\vect{u}_i,
\end{IEEEeqnarray*}
where $\vect{u}_i = ({u_{i,1}, u_{i,2}, \dots, u_{i,n}}) \in \Integers_2^n$, and concatenate all such elements in a vector $(\vect{u}_1, \vect{u}_2, \dots, \vect{u}_L) \in \Integers_2^{nL}$. By running through all codewords $\vect{c} \in \code{C}$, and placing all the binary codewords in a set $\code{B} \subseteq \mathbb{Z}_2^{nL}$, we get the (binary) decomposition code $\code{B}$.
\end{definition}

The relation between decomposition codes and lattices was first presented in~\cite{maiara}.

\begin{example}\label{ex:linear_concantenated} Consider $\code{C} = \{(0,0),(1,3),(2,2),(3,1)\} \subset \Integers_4^2$. Then 
\begin{IEEEeqnarray*}{rCl}
(0,0) = (0,0) + 2(0,0) \mapsto (0,0,0,0),~ ~ ~ ~ & (1,3) = (1,1) + 2(0,1) \mapsto (1,1,0,1), \\
(2,2) = (0,0) + 2(1,1) \mapsto (0,0,1,1),~ ~ ~ ~ & (3,1) = (1,1) + 2(1,0) \mapsto (1,1,1,0).
\end{IEEEeqnarray*}
The decomposition code $\code{B}=\{(0,0,0,0),(1,1,0,1),(0,0,1,1),(1,1,1,0)\} \subseteq \Integers_2^4$ is linear, but this property does not hold in general.
\hfill\exampleend
\end{example}

We will now demonstrate a necessary and sufficient condition, based on the decomposition code, that guarantees the linearity of $\gray(\code{C})$. This result simply translates the condition for the $\Integers_{2^L}$-additive code $\code{C}$ on Theorem~\ref{tapiavega} to the decomposition code $\code{B}$.

\begin{theorem} 
\label{thm:iff_condition_concatenated} 
Consider a $\Integers_{2^{L}}$-additive code $\code{C}$ and its respective binary decomposition code $\code{B}$. Then $\gray(\code{C})$ is linear if and only if 
\begin{IEEEeqnarray*}{c}
\set{R}=\left\{(\vect{0},\vect{u}_1 \circ \vect{v}_1, \vect{u}_2 \circ \vect{v}_2, \dots, \vect{u}_{L-1} \circ \vect{v}_{L-1}): \forall ~\vect{c} = \sum_{i=1}^L 2^{i-1}\vect{u}_i,~ \vect{d} = \sum_{i=1}^L 2^{i-1}\vect{v}_i \in \code{C} \right\}
\end{IEEEeqnarray*}
is a subset of $\code{B}$.
\end{theorem}

\begin{proof}
Let $\code{C} \subseteq \mathbb{Z}_{2^L} ^n$ be a $\mathbb{Z}_{2^L}$-additive code and $\vect{c},\vect{d} \in \code{C}$ where
\begin{IEEEeqnarray}{rCl}
\label{eq:u-v-expression}
\vect{c}&=&\left(\sum_{i=1}^{L} 2^{i-1}u_{i,1}, \sum_{i=1}^{L} 2^{i-1}u_{i,2},\ldots, \sum_{i=1}^{L} 2^{i-1}  u_{i,n}  \right) = \sum_{i=1}^L 2^{i-1}\vect{u}_i, \nonumber \\ 
\vect{d}&=&\left(\sum_{i=1}^{L} 2^{i-1}v_{i,1}, \sum_{i=1}^{L} 2^{i-1}v_{i,2},\ldots, \sum_{i=1}^{L} 2^{i-1}v_{i,n} \right) = \sum_{i=1}^L 2^{i-1}\vect{v}_i ,
\end{IEEEeqnarray}
where $\vect{u}_i = ({u_{i,1}, u_{i,2}, \dots, u_{i,n}})$ and $\vect{v}_i = ({v_{i,1}, v_{i,2}, \dots, v_{i,n}})$, for $i\in [1: L]$. Consider also the decomposition code $\code{B}$ obtained from $\code{C}$. Rewriting $\vect{c}\odot \vect{d}$, we get
\begin{IEEEeqnarray}{rCl}\label{eq:expansion_odot}
\vect{c}\odot \vect{d} & = & \bigg(\sum_{i=1}^{L} 2^{i-1} (u_{i,1} \odotq{2} v_{i,1}), \sum_{i=1}^{L} 2^{i-1} (u_{i,2} \odotq{2} v_{i,2}) , \dots, \sum_{i=1}^{L} 2^{i-1} (u_{i,n} \odotq{2} v_{i,n})  \bigg) \nonumber \\
& = & (u_{1,1} \odotq{2} v_{1,1}, u_{1,2} \odotq{2} v_{1,2}, \dots, u_{1,n} \odotq{2} v_{1,n} ) +  \nonumber \\
& & 2(u_{2,1} \odotq{2} v_{2,1},u_{2,2} \odotq{2} v_{2,2}, \dots, u_{2,n} \odotq{2} v_{2,n})  + \dots + \nonumber \\
& & 2^{L-1} (u_{L,1} \odotq{2} v_{L,1},u_{L,2} \odotq{2} v_{L,2}, \dots, u_{L,n} \odotq{2} v_{L,n}).
\end{IEEEeqnarray}

Hence, each term of~\eqref{eq:expansion_odot} reduces to $\vect{u}_i \circ \vect{v}_i$ for $i\in [1: L]$, in terms of the Schur product, since the operations $\odotq{2}$ and $\circ$ coincide. Also,
\begin{IEEEeqnarray}{rCl}
\label{eq:2uv}
2(\vect{c}\odot \vect{d}) & = &  2 (\vect{u}_1 \circ \vect{v}_1) + 2^2 (\vect{u}_2 \circ \vect{v}_2) + \dots + 2^{L-1}(\vect{u}_{L-1} \circ \vect{v}_{L-1}).
\end{IEEEeqnarray}

The following  equivalence of statements holds:
\begin{enumerate}
    \item[1)] $(\vect{0},\vect{u}_1 \circ \vect{v}_1, \vect{u}_2 \circ \vect{v}_2, \dots, \vect{u}_{L-1} \circ \vect{v}_{L-1}) \in \code{B}$
    \item[2)] $2(\vect{c}\odot \vect{d}) \in \code{C}.$
\end{enumerate}

To demonstrate 1) $\Rightarrow$ 2), given that $(\vect{0},\vect{u}_1 \circ \vect{v}_1, \vect{u}_2 \circ \vect{v}_2, \dots, \vect{u}_{L-1} \circ \vect{v}_{L-1}) \in \code{B}$, then there exist $\vect{x} = \vect{0} + 2(\vect{u}_1 \circ \vect{v}_1) + 2^2 (\vect{u}_2 \circ \vect{v}_2) + \dots + 2^{L-1}(\vect{u}_{L-1} \circ \vect{v}_{L-1}) = 2(\vect{c}\odot \vect{d}) \in \code{C}$. On the other hand, 2) $\Rightarrow$ 1) follows directly from the definition of a decomposition code. This equivalence, together with Theorem~\ref{tapiavega}, completes the proof.
\end{proof}

For the case of $L=2$, i.e., when we restrict Theorem~\ref{thm:iff_condition_concatenated} to $\Integers_4$-additive codes, then our result coincides with~\cite[Th.~5]{Hammons_etal_94}, which claims that
\begin{IEEEeqnarray*}{c}
\gray(\code{C})~\textnormal{is linear} \Leftrightarrow 2\left( (\vect{c} \bmod 2) \circ (\vect{d} \bmod 2) \right) \in \code{C}, 
\end{IEEEeqnarray*}
for all $\vect{c}, \vect{d} \in \code{C}$, where $\code{C}$ is a $\Integers_4$-additive code. Also,
\begin{IEEEeqnarray}{c}
2\left( (\vect{c} \bmod 2) \circ (\vect{d} \bmod 2) \right) \in \code{C} \Leftrightarrow (\vect{0}, \vect{u}_1 \circ \vect{v}_1) \in \code{B},
\end{IEEEeqnarray}
since $\vect{c} = \vect{u}_1 + 2\vect{u}_2$ and $\vect{d} = \vect{v}_1 + 2\vect{v}_2$, we have that $\vect{u}_1 = \vect{c} \bmod 2$ and $\vect{v}_1 = \vect{d} \bmod 2$. The result then follows from the definition of decomposition code. In this sense, Theorem~\ref{thm:iff_condition_concatenated} can be seen as an extension and/or generalization of~\cite[Th.~5]{Hammons_etal_94} to $\Integers_{2^L}$-additive codes with $L>2$.

We restate our result in terms of the Schur product $\code{C}_i \circ \code{C}_i$, $i \in [1: L-1]$.

\begin{proposition} 
\label{prop:connections-hammons-etal}
Consider a $\Integers_{2^{L}}$-additive code $\code{C}$. Then $\gray(\code{C})$ is linear if and only if $2 (\code{C}_1 \circ \code{C}_1) + 4 (\code{C}_2
\circ \code{C}_2) + \dots + 2^{L-1} (\code{C}_{L-1} \circ \code{C}_{L-1}) \in \code{C}$.
\end{proposition}

\begin{proof} Consider all $\vect{c}, \vect{d} \in 
\code{C}$, where $\vect{c} = \sum_{i=1}^L 2^{i-1}\vect{u}_i, \vect{d} = \sum_{i=1}^L 2^{i-1}\vect{v}_i \in \code{C}$. The proof follows from the fact that $(\vect{0},\vect{u}_1 \circ \vect{v}_1, \vect{u}_2 \circ \vect{v}_2, \dots, \vect{u}_{L-1} \circ \vect{v}_{L-1}) \in \code{B}$, where $\code{B}$ is the decomposition code, is equivalent to $2(\vect{u}_1 \circ \vect{v}_1) + 4(\vect{u}_2 \circ \vect{v}_2) + \dots + 2^{L-1}(\vect{u}_{L-1} \circ \vect{v}_{L-1}) \in \code{C}$. Notice that $\vect{u}_i \circ \vect{v}_i \in \code{C}_i \circ \code{C}_i$ for all $i \in [1: L-1]$.
\end{proof}

\subsection{Special Cases: $\Integers_4$ and $\Integers_8$-Additive Codes}
\label{subsec:z4-and-z8}

From the decomposition code and Theorem~\ref{thm:iff_condition_concatenated}, we introduce a simplified way of verifying the linearity of $\gray(\code{C})$ for $\Integers_{4}$ and $\Integers_8$-additive codes. Consider $(\vect{u}_1,\vect{u}_2,\dots,\vect{u}_L), (\vect{v}_1,\vect{v}_2,\dots,\vect{v}_L) \in \code{B}$. We define
\begin{IEEEeqnarray}{c}
\label{eq:definition-set-S}
\vect{s}_i = (\vect{u}_i \circ \vect{v}_i) \oplusq{2} \vect{r}_i^{(1)} \oplusq{2} \dots \oplusq{2}  \vect{r}_{i}^{(i-1)}, \nonumber \\
 \vect{r}_i^{(1)} = (\vect{u}_i \oplusq{2} \vect{v}_i) \circ (\vect{u}_{i-1} \circ \vect{v}_{i-1}), ~ \vect{r}_i^{(j)} = \vect{r}_i^{(j-1)} \circ \vect{r}_{i-1}^{(j-1)},
\end{IEEEeqnarray}
for $2 \leq j \leq i-1,$ $2 \leq i \leq L-1$, $\vect{s}_0=\vect{0}$. Thus, we have for $i=1,2$ that
\begin{IEEEeqnarray}{rCl}
\label{eq:formulas-si}
\vect{s}_1 & = & \vect{u}_1 \circ \vect{v}_1, \nonumber \\
\vect{s}_2 & = &  \left(\vect{s}_1 \circ (\vect{u}_2 \oplusq{2} \vect{v}_2) \right) \oplusq{2} (\vect{u}_2 \circ \vect{v}_2) = \left((\vect{u}_1 \circ \vect{v}_1) \circ (\vect{u}_2 \oplusq{2} \vect{v}_2) \right) \oplusq{2} (\vect{u}_2 \circ \vect{v}_2).
\end{IEEEeqnarray}

When the decomposition code $\code{B}$ is linear, we can derive some properties about the set $\mathcal{S}=\{(\vect{0},\vect{s}_1,\dots,\vect{s}_{L-1})\}$, where each $\vect{s}_i$ is defined as in~\eqref{eq:definition-set-S}, for all $\vect{u},\vect{v} \in \code{B}$.

\begin{theorem}\cite[Th.~5]{maiara} Let $\code{B}$ be a linear code in $\Integers_2^{nL}$ obtained from a code $\code{C} \subseteq \Integers_{2^L}^n$. Then, $\code{C}$ is $\Integers_{2^L}$-additive if and only if $\mathcal{S} \subseteq \code{B}$.
\end{theorem}

The linearity of $\gray(\code{C})$ is related to the linearity of the decomposition code $\code{B}$ for $\Integers_4$ and $\Integers_8$-additive codes. 

\begin{lemma} 
\label{lem:decomposition-gray-linear}
Let $\code{C}$ be $\Integers_{2^L}$-additive, for $L=2,3$. The decomposition code $\code{B}$ is linear if and only if $\gray(\code{C})$ is linear.
\end{lemma}

\begin{proof} $(\Rightarrow)$ We start with $L=2$, i.e. we assume that $\code{C}$ is $\Integers_4$-additive. If $\code{B}$ is linear, it follows immediately that $\gray(\code{C})$ is also linear from Theorem~\ref{thm:iff_condition_concatenated} since $\set{R}=\set{S}$. 

For $L=3$, assume $\code{B}$ is linear and fix $\vect{c} =\vect{u}_1+2\vect{u}_2+4\vect{u}_3, \vect{d} =\vect{v}_1+2\vect{v}_2+4\vect{v}_3 \in \code{C} \subseteq \Integers_8^n$, and respectively $(\vect{u}_1, \vect{u}_2, \vect{u}_3), (\vect{v}_1, \vect{v}_2, \vect{v}_3) \in \code{B} \subseteq \Integers_{2}^{3n}$. Since $\code{B}$ is linear, $(\vect{u}_1 \oplusq{2} \vect{v}_1, \vect{u}_2 \oplusq{2} \vect{v}_2, \vect{u}_3 \oplusq{2} \vect{v}_3) \in \code{B}$ and $(\vect{u}_1 \oplusq{2} \vect{v}_1) +2 (\vect{u}_2 \oplusq{2} \vect{v}_2) + 4 (\vect{u}_3 \oplusq{2} \vect{v}_3) \in \code{C}$. From Theorem~\ref{thm:iff_condition_concatenated}, $\set{S}=\{(0, \vect{s}_1, \vect{s}_2)\} \subseteq \code{B}$, where $\vect{s}_1, \vect{s}_2$ are defined as in~\eqref{eq:formulas-si}, and $2 \vect{s}_1 + 4 \vect{s}_2 \in \code{C}$. Since $\code{C}$ is $\Integers_8$-additive,
\begin{IEEEeqnarray}{c}
\label{eq:linearity-z8}
\bigl( (\vect{u}_1 \oplusq{2} \vect{v}_1) +2 (\vect{u}_2 \oplusq{2} \vect{v}_2) + 4 (\vect{u}_3 \oplusq{2} \vect{v}_3) \bigr) \oplusq{8} \bigl(2 \vect{s}_1 + 4 \vect{s}_2 \bigr) \in \code{C} \nonumber \\
\Rightarrow (\vect{u}_1 \oplusq{2} \vect{v}_1) + 2 ((\vect{u}_2 \oplusq{2} \vect{v}_2) \oplusq{2} \vect{s}_1) + 4 \bigl(((\vect{u}_2 \oplusq{2} \vect{v}_2) \circ \vect{s}_1) \oplusq{2} (\vect{u}_3 \oplusq{2} \vect{v}_3) \oplusq{2} \vect{s}_2 \bigr) \in \code{C}  \nonumber \\
\Rightarrow \bigl(\vect{u}_1 \oplusq{2} \vect{v}_1, (\vect{u}_2 \oplusq{2} \vect{v}_2) \oplusq{2} \vect{s}_1,  ((\vect{u}_2 \oplusq{2} \vect{v}_2) \circ \vect{s}_1) \oplusq{2} (\vect{u}_3 \oplusq{2} \vect{v}_3) \oplusq{2} \vect{s}_2 \bigr) \in \code{B},
\end{IEEEeqnarray}
and $\bigl(\vect{0}, \vect{0},  (\vect{u}_2 \oplusq{2} \vect{v}_2) \circ \vect{s}_1 \bigr) \in \code{B}$. Provided that $\vect{s}_1 = \vect{u}_1 \circ \vect{v}_1$, $\vect{s}_2 = ((\vect{u}_2 \oplusq{2} \vect{v}_2) \circ \vect{s}_1 ) \oplusq{2} (\vect{u}_2 \circ \vect{v}_2)$, and $(\vect{0},\vect{s}_1, \vect{s}_2) \in \code{B}$, we get
\begin{IEEEeqnarray*}{c}
(\vect{0},\vect{s}_1, \vect{s}_2) = (\vect{0}, \vect{u}_1 \circ \vect{v}_1, \vect{u}_2 \circ \vect{v}_2) \oplusq{2} (\vect{0}, \vect{0}, (\vect{u}_2 \oplusq{2} \vect{v}_2) \circ \vect{s}_1) \Rightarrow (\vect{0}, \vect{u}_1 \circ \vect{v}_1, \vect{u}_2 \circ \vect{v}_2) \in \code{B},
\end{IEEEeqnarray*}
and from Theorem~\ref{thm:iff_condition_concatenated}, $\gray(\code{C})$ is linear, as we wanted to demonstrate.

$(\Leftarrow)$ For $L=2$, if $\gray(\code{C})$ is linear then $\set{R}=\{(0, \vect{u}_1 \circ \vect{v}_1): \forall \vect{c}=\vect{u}_1+2\vect{u}_2, ~\vect{d}=\vect{v}_1+2\vect{v}_2 \in \code{C} \} \subseteq \code{B}$ from Theorem~\ref{thm:iff_condition_concatenated}, which implies that $(0,\vect{s}_1) \in \code{B}$. Thus, given $\vect{c}=\vect{u}_1+2\vect{u}_2$, $\vect{d}=\vect{v}_1+2\vect{v}_2 \in \code{C}$, such that $(\vect{u}_1, \vect{u}_2), (\vect{v}_1, \vect{v}_2) \in \code{B}$, it holds that
\begin{IEEEeqnarray*}{rCl}
\vect{c} \oplusq{4} \vect{d} \in \code{C}  & \Rightarrow & (\vect{u}_1 \oplusq{2} \vect{v}_1) + 2\bigl((\vect{u}_1 \circ \vect{v}_1) \oplusq{2} (\vect{u}_2 \oplusq{2} \vect{v}_2) \bigr) \\
& \Rightarrow & \bigl((\vect{u}_1 \oplusq{2} \vect{v}_1), (\vect{u}_1 \circ \vect{v}_1) \oplusq{2} (\vect{u}_2 \oplusq{2} \vect{v}_2)  \bigr) \in \code{B},
\end{IEEEeqnarray*}
and we can conclude that $\bigl( \vect{u}_1 \oplusq{2} \vect{v}_1, \vect{u}_2 \oplusq{2} \vect{v}_2 \bigr) \in \code{B}$, since $(0, \vect{u}_1 \circ \vect{v}_1) \in \code{B}$. Given that $\vect{0} \in \code{C}$, also $\vect{0} \in \code{B}$, therefore $\code{B}$ is linear.

When $L=3$ and $\gray(\code{C})$ is assumed to be linear, Theorem~\ref{thm:iff_condition_concatenated} guarantees that 
\begin{IEEEeqnarray}{c}
\label{eq:set-r-3levels}
\set{R} =\left\{(\vect{0},\vect{u}_1 \circ \vect{v}_1, \vect{u}_2 \circ \vect{v}_2): \forall ~\vect{c}=\vect{u}_1+2\vect{u}_2+4\vect{u}_3, \vect{d}=\vect{v}_1+2\vect{v}_2+4\vect{v}_3 \in \code{C} \right\} \subseteq \code{B}. \nonumber \\
\end{IEEEeqnarray}
Also, since $\code{C}$ is $\Integers_8$-additive, for $\vect{c}, \vect{d} \in \code{C}$, it holds that $\vect{c} \oplusq{8} \vect{d} \in \code{C}$ and thus
\begin{IEEEeqnarray}{c}
\label{eq:linearity-c-3levels}
\vect{u}_1 \oplusq{2} \vect{v}_1 + 2\bigr((\vect{u}_2 \oplusq{2} \vect{v}_2) \oplusq{2} (\vect{u}_1 \circ \vect{v}_1) \bigl) \nonumber \\
+ 4  \bigr(((\vect{u}_2 \oplusq{2} \vect{v}_2) \circ (\vect{u}_1 \circ \vect{v}_1)) \oplusq{2} (\vect{u}_3 \oplusq{2} \vect{v}_3) \oplusq{2} (\vect{u}_2 \circ \vect{v}_2) \bigr) \in \code{C}.
\end{IEEEeqnarray}
 We want to demonstrate that $\code{B}$ is linear. Clearly, $\vect{0} \in \code{B}$. By adding the corresponding vectors from~\eqref{eq:set-r-3levels} and~\eqref{eq:linearity-c-3levels} in $\code{C}$, it follows that $(\vect{u}_1 \oplusq{2} \vect{v}_1, \vect{u}_2 \oplusq{2} \vect{v}_2, \vect{u}_3 \oplusq{2} \vect{v}_3) \in \code{B}$.
\end{proof}

Checking the linearity of the decomposition code, rather than the $\Integers_{2^L}$-linear code itself is advantageous since the former has a simpler structure. For example, if $\code{B}$ is linear, then its basis follows directly from the $\Integers_4$-additive codes. Given the relationship between the decomposition code $\code{B}$ and the respective $\Integers_{2^L}$-linear code via Lemma~\ref{lem:decomposition-gray-linear}, finding a basis for $\code{B}$ is also useful for constructing linear $\Integers_{2^L}$-linear codes.


Recall that any $\Integers_4$-additive code $\code{C}$ is permutation equivalent to a code with a generator matrix $\mat{G}_{\code{C}}$ in the standard form
\begin{IEEEeqnarray}{c}
\label{eq:generator_cs}
  \mat{G}_{\code{C}}=\begin{pNiceMatrix}
    \mat{I}_{k_1} & \mat{A} & \mat{B} \\
    \mat{O}_{k_2 \times k_1} & 2\mat{I}_{k_2} & 2\mat{C}
    \label{eq:Generator_Matrix}
  \end{pNiceMatrix},
\end{IEEEeqnarray}
where $\mat{A}$ and $\mat{C}$ are binary matrices, and $\mat{B}$ is defined over $\Integers_4$~\cite[Prop.~1.1]{Wan97_1}. Thus, there are $k=k_1+k_2$ generator vectors of $\code{C}$ and the cardinality of $\code{C}$ is $4^{k_1}2^{k_2}$ .

\begin{lemma} 
\label{lemm:basis-z4}
Let $\code{C}$ be a $\Integers_{4}$-additive code with generator vectors $\vect{c}_1, \vect{c}_2, \dots, \vect{c}_k$. If the decomposition code $\code{B}$ is linear, then $\gamma=\{(\vect{u}_{i,1},\vect{u}_{i,2}),(\vect{0},\vect{u}_{i,1}): 1\leq i \leq k\}$
is a basis of $\code{B}$, where $\vect{c}_i=\vect{u}_{i,1}+2\vect{u}_{i,2}$, for all $i$.   
\end{lemma}

\begin{proof} Assume the generator matrix of the $\Integers_4$-additive code as in~\eqref{eq:generator_cs}. Thus,
\begin{enumerate}
    \item[i)] \emph{The $2k_1+k_2$ elements in $\gamma$ are linearly independent.}\\
    Indeed, the binary decomposition of the first set of rows of $\mat{G}_{\code{C}}$ gives $k_1$ vectors $(\vect{u}_{i,1},\vect{u}_{i,2})$, for $i \in [1:k_1]$, while from the second set of rows we get $k_2$ vectors of the form $(\vect{0},\vect{u}_{j,2})$, for $j \in [1: k_2]$. These two classes of vectors are clearly linearly independent. Considering the set $\gamma$, also the $k_1$ vectors $(\vect{0},\vect{u}_{i,1})$ should compose the basis. Notice that $(\vect{0},\vect{u}_{i,1})$ and $(\vect{0},\vect{u}_{j,2})$ are linearly independent, since $\vect{u}_{j,2}$ has the initial $k_1$ coordinates equal to zero, where $\vect{u}_{i,1}$ has the canonical vectors originating from the block $\mat{I}_{k_1}$. Thus, the $2k_1+k_2$ vectors in $\gamma$ are linearly independent.

    \item[ii)] \emph{The vectors in $\gamma$ generate $\code{B}$.} \\
    Consider an element $(\vect{u}_1,\vect{u}_2) \in \code{B}$. Then, $\vect{c}=\vect{u}_1+2\vect{u}_2 \in \code{C}$ and also 
    \begin{IEEEeqnarray}{c}
    \label{eq:linear-combination}
    \vect{c}=\sum_{\ell=1}^{k} \alpha_\ell \vect{c}_\ell = \sum_{\ell=1}^{k} \alpha_\ell (\vect{u}_{\ell,1}+2\vect{u}_{\ell,2}),
    \end{IEEEeqnarray}
    where $\alpha_{\ell} \in \Integers_4$. For each choice of $\alpha_{\ell}$, observe that
    \begin{IEEEeqnarray*}{rCl}
    & 0 \cdot (\vect{u}_{\ell,1}+2\vect{u}_{\ell,2}) \mapsto (\vect{0},\vect{0}), ~ ~ ~ ~ & 1 \cdot (\vect{u}_{\ell,1}+2\vect{u}_{\ell,2}) \mapsto (\vect{u}_{\ell,1},\vect{u}_{\ell,2}), \\
    & 2 \cdot (\vect{u}_{\ell,1}+2\vect{u}_{\ell,2}) \mapsto (\vect{0},\vect{u}_{\ell,1}), ~ ~ ~ ~ & 3 \cdot (\vect{u}_{\ell,1}+2\vect{u}_{\ell,2}) \mapsto (\vect{u}_{\ell,1},\vect{u}_{\ell,1} \oplusq{2} \vect{u}_{\ell,2})
    \end{IEEEeqnarray*}
    and these four vectors are a linear combination of elements in $\gamma$.

    From~\eqref{eq:linear-combination} and $\lambda_i=|\{\alpha_\ell=i, 1\leq \ell \leq k\}|$, for $i \in [0:3]$ in the composition of $\vect{c}$,
    \begin{IEEEeqnarray*}{rCl}
    (\vect{u}_1,\vect{u}_2) & = &  \lambda_{0}(\vect{0},\vect{0})+ \lambda_{1}(\vect{u}_{\ell,1},\vect{u}_{\ell,2}) + \lambda_2 (\vect{0},\vect{u}_{\ell,1}) + \lambda_3 (\vect{u}_{\ell,1},\vect{u}_{\ell,1} \oplusq{2} \vect{u}_{\ell,2}) \\
    & = & \left( (\lambda_1 + \lambda_3) \vect{u}_{\ell,1}, \lambda_1 \vect{u}_{\ell,2} + \lambda_2 \vect{u}_{\ell,1} + \lambda_3 (\vect{u}_{\ell,1} \oplusq{2} \vect{u}_{\ell,2})   \right).
    \end{IEEEeqnarray*}
    Notice that
     \begin{IEEEeqnarray*}{rCl}
    (\vect{u}_1,\vect{u}_2) = 
    \begin{cases}
    (\vect{0}, \vect{0}), & \text{if} ~ ~ \lambda_1, \lambda_2, \lambda_3 \equiv 0 \bmod 2 \\
    & \text{or} ~ ~ \lambda_1, \lambda_2, \lambda_3 \equiv 1 \bmod 2, \\
    (\vect{u}_{\ell,1}, \vect{u}_{\ell,1} \oplusq{2} \vect{u}_{\ell,2}),  &  \text{if} ~ ~\lambda_1, \lambda_2 \equiv 0 \bmod 2 ~ \text{and} ~ \lambda_3 \equiv 1 \bmod 2  \\
     &  \text{or} ~ ~\lambda_1, \lambda_2 \equiv 1 \bmod 2 ~ \text{and} ~ \lambda_3 \equiv 0 \bmod 2,  \\
     (\vect{0}, \vect{u}_{\ell,1}),  & \text{if} ~ ~ \lambda_1, \lambda_3 \equiv 0 \bmod 2 ~ \text{and} ~ \lambda_2 \equiv 1 \bmod 2 \\
     & \text{or} ~ ~\lambda_1, \lambda_3 \equiv 1 \bmod 2 ~ \text{and} ~ \lambda_2 \equiv 0 \bmod 2,  \\
    (\vect{u}_{\ell,1}, \vect{u}_{\ell,2}),  &  \text{if} ~ ~\lambda_2, \lambda_3 \equiv 0 \bmod 2 ~ \text{and} ~ \lambda_1 \equiv 1 \bmod 2  \\
    &  \text{or} ~ ~\lambda_2, \lambda_3 \equiv 1 \bmod 2 ~ \text{and} ~ \lambda_1 \equiv 0 \bmod 2,  \\
    \end{cases}
    \end{IEEEeqnarray*}
which are clearly generated by the vectors in $\gamma$. This concludes the proof. \hfill \qedhere
\end{enumerate}
\end{proof}

The following result gives a simple sufficient (but not necessary) condition for a $\Integers_4$-linear code to be linear based on its generator matrix in the standard form~\eqref{eq:generator_cs}.

\begin{lemma}
\label{lem:simple-condition}
Consider a $\Integers_{4}$-additive code $\code{C}$. Then $\gray(\code{C})$ is linear if $\code{C}$ has a generator matrix $\mat{G}_{\code{C}}$ (say, on the form \eqref{eq:generator_cs})  in which no column contains more than one odd-valued coordinate. 
\end{lemma}

\begin{proof} For $\Integers_{4}$-additive codes, the Gray map preserves addition from $\Integers_4$ to $\Integers_2$ except when adding two odd numbers, which can be easily verified by checking all sums of two numbers in $\Integers_{4}$.
\end{proof}

To indicate the connection between Lemma~\ref{lem:simple-condition} and Theorem~\ref{thm:iff_condition_concatenated}, observe that the conditions of the lemma induce a partitioning of the coordinates, in the sense that an odd value occurring in coordinate $i$ ``belongs'' to a specific row $\vect{r}$ of the generator matrix $\mat{G}_{\code{C}}$. Hence, any codeword $\vect{c}$ with an odd value in coordinate $i$ is a linear combination of rows with an odd-valued multiple of row $\vect{r}$. It follows that all Schur products on the form  $\vect{u}_1 \circ \vect{v}_1$ are actually codewords in $\code{C}_1$.

\begin{remark}
\label{rem:simplegenmat}
    In fact, a generator matrix that meets the condition in Lemma~\ref{lem:simple-condition} can be straightforwardly expanded into a generator matrix for $\gray(\code{C})$  by replacing
    \begin{itemize}
        \item each row $\vect{r}$ in $\mat{G}_{\code{C}}$ that contains at least one odd-numbered value by two rows corresponding to the Gray map vectors  $\gray(\vect{r})$ and  $\gray(\vect{3r})$, respectively, and 
        \item each row $\vect{r}$ in $\mat{G}_{\code{C}}$ that contains only zeros and twos by the  Gray map vector  $\gray(\vect{r})$.
    \end{itemize}
\end{remark}

\begin{example}
As it can be observed in Table~\ref{tab:table-gray-linearity}, restricting code search to those that meet Lemma~\ref{lem:simple-condition} may have a cost in terms of optimum Lee distance. However, searches for short codes with good Euclidean distance\footnote{Recall that the Euclidean distance between two vectors $\vect{x},\vect{y}\in \Reals^n$ is $d_{\textnormal{E}}(\vect{x},\vect{y})=\sqrt{\sum_{i=1}^n (y_i-x_i)^2}$.} suggest that codes with a linear Gray map may perform well.  The $[10,4^5]$ $\Integers_4$-additive code $\code{C}$ generated by 
\begin{IEEEeqnarray*}{c}
   \mat{G}_{\code{C}}=\begin{pNiceMatrix}
 1& 0& 0& 0& 0& 0& 0& 1& 2 & 2\\
 0& 1& 0& 0& 0& 0& 1& 2& 0 & 2\\
 0& 0& 1& 0& 0& 1& 2& 2& 0 & 0\\
 0& 0& 0& 1& 0& 2& 0& 0& 2 & 1\\
 0& 0& 0& 0& 1& 2& 2& 0& 1 & 0\\
\end{pNiceMatrix}
\label{eq:G10-5}
\end{IEEEeqnarray*}   
has the largest squared Euclidean distance ($=8$) among all linear codes of length $10$ and cardinality $4^5$ over $\Integers_{4}$. By Lemma~\ref{lem:simple-condition} its Gray map code $\gray(\code{C})$ is a binary linear $[20,2^{10}]$ code. Its binary generator matrix is obtained as indicated in Remark~\ref{rem:simplegenmat}.
\end{example}

\subsection{Nonlinearity}
\label{subsec:nonlinearity}

To demonstrate the nonlinearity of some binary codes obtained via Gray map, we consider the associated codes instead of the decomposition codes. The advantages of this contribution are threefold: i) only binary operations are performed, ii) the condition can fail on the initial levels, and iii) if the condition fails for $\overline{n} \ll n$ coordinates, then clearly the code is nonlinear. 

\begin{corollary}\label{coro:princ}
Let $\code{C} \subseteq \mathbb{Z}_{2^L} ^n$ be a $\mathbb{Z}_{2^L}$-additive code with associated codes $\code{C}_1 ,\code{C}_2, \ldots, \code{C}_L$. If $\gray(\code{C})$ is a linear code, then the associated codes are closed under Schur product, i.e., $\code{C}_i \circ \code{C}_i \subseteq \code{C}_{i+1}$ for all $i \in [1: L-1]$.
\end{corollary}

\begin{proof} From Theorem~\ref{thm:iff_condition_concatenated}, given that $\gray(\code{C})$ is linear, then $(\vect{0},\vect{u}_1 \circ \vect{v}_1, \vect{u}_2 \circ \vect{v}_2, \dots, \vect{v}_{L-1} \circ \vect{v}_{L-1}) \in \code{B}$, where $\code{B}$ is the decomposition code obtained from $\code{C}$. In terms of the associated codes, it implies that $\vect{0} \in \code{C}_1$, $\vect{u}_1 \circ \vect{v}_1 \in \code{C}_2$, and in more general terms, $\vect{u}_i \circ \vect{v}_i \in \code{C}_{i+1}$, for all $\vect{c}, \vect{d} \in \code{C}$ as in~\eqref{eq:u-v-expression}, meaning that the associated codes are closed under Schur product, as we wanted to demonstrate. 
\end{proof}

Observe that nonlinear codes are of particular interest in the context of Gray map, as there exist nonlinear codes with remarkable properties that could not be achieved by linear ones of the same length, such as the Nordstrom-Robinson code~\cite{NordstromRobinson_67}.

\begin{example}
\label{ex:nordstrom-robinson-condition}
We apply Corollary~\ref{coro:princ} to show the nonlinearity of the Nordstrom-Robinson code. Indeed, not all vectors of the form $(\vect{0},\vect{u}_1 \circ \vect{v}_1)$ belong to the decomposition code $\code{B}$. Consider the codewords $\vect{c}=(0,0,0,1,2,3,1,1), \vect{d}=(0,1,0,0,1,2,3,1) \in \code{O} \subseteq \Integers_4^8$. When we perform the binary decomposition of these vectors, we get
\begin{IEEEeqnarray*}{c}
\vect{u}=(\underbrace{0,0,0,1,0,1,1,1}_{\vect{u}_1},\underbrace{0,0,0,0,1,1,0,0}_{\vect{u}_2})\mbox{ and }\vect{v}=(\underbrace{0,1,0,0,1,0,1,1}_{\vect{v}_1},\underbrace{0,0,0,0,0,1,1,0}_{\vect{v}_2}) 
\end{IEEEeqnarray*}
which are the respective codewords in the decomposition code $\code{B}$. Nonetheless, 
\begin{IEEEeqnarray*}{c}
(\vect{0},\vect{u}_1 \circ \vect{v}_1)=(0,0,0,0,0,0,0,0,\underbrace{0,0,0,0,0,0,1,1}_{\vect{u}_1 \circ \vect{v}_1}) \notin \code{B},
\end{IEEEeqnarray*}
since its respective codeword $(0,0,0,0,0,0,2,2) \notin \code{O}$. Hence, from Theorem~\ref{thm:iff_condition_concatenated}, the Nordstrom-Robinson code is nonlinear. We have computationally checked that $\code{C}_1 \circ \code{C}_1 \subseteq \code{C}_2$, meaning that the necessary condition of Corollary~\ref{coro:princ} does not hold.
\hfill\exampleend
\end{example} 

\subsection{Distance Properties}
\label{sec:distance-properties}
Much of the practical interest in codes over rings is related to the fact that for a $\Integers_{4}$-linear code, its minimum Lee distance equals the minimum Hamming distance of its Gray map code, but this is not true  in general for larger rings. If some kind of distance preservation could be guaranteed, 
it could be useful for constructing large nonlinear binary codes with good minimum distance properties.
In this section, we investigate the relationship between the Lee distance of a $\Integers_{2^L}$-additive code and the Hamming distance of its respective $\Integers_{2^L}$-linear code.

Given $\vect{c},\vect{d} \in \code{C}$, where $\code{C}$ is a $\Integers_{2^L}$-additive, we have $\Lwt{\vect{c}-\vect{d}} = d_{\textnormal{Lee}}(\vect{c},\vect{d})$. 
Recall that for the Hamming distance between two elements, it holds that $d_{\textnormal{H}}(\gray(\vect{c}),\gray(\vect{d}))=w_{\textnormal{H}}(\gray(\vect{c}-\vect{d}))$, for all $\vect{c}, \vect{d}$ in the $\Integers_{2^L}$-additive code $\code{C}$~\cite[Prop.~1]{carlet}.

\begin{proposition}
\label{prop:lee-hammimng-weight}
Consider $\vect{c}=(c_1, c_2, \cdots, c_n) \in \code{C} \subseteq \Integers_{2^L}^n$, then
\begin{IEEEeqnarray*}{rCl}
& \Lwt{\vect{c}} & = n + (2^{L-1}-1)\delta_{2^{L-1}} - \delta_0 + \displaystyle{\sum_{i=1}^{2^{L-1}-1}}(i-1)(\delta_i + \delta_{2^L -i}), ~ ~ ~ \text{and} \\
& w_{\textnormal{H}}(\gray(\vect{c}))& = 2^{L-2}(n + \delta_{2^{L-1}}- \delta_0),
\end{IEEEeqnarray*}
where $\delta_i=|\{c_\ell=i, 1\leq \ell \leq n\}|$, for $i \in [1: 2^{L}-1]$. 
\end{proposition}

\begin{proof}
Let $\vect{c}=(c_1, c_2, \cdots, c_n) \in \code{C} \subset \Integers_{2^L}^n$, then $\min\{c_i, 2^L -c_i\}=k$, if $c_i=k$ or $c_i=2^L -k$. 
Hence,
\begin{IEEEeqnarray*}{rCl}
\Lwt{\vect{c}} & = & \displaystyle{\sum_{i=1}^{n}}\min\{c_i, 2^L -c_i\}= 2^{L-1}\delta_{2^{L-1}} + \displaystyle{\sum_{i=1}^{2^{L-1}-1}}i(\delta_i + \delta_{2^L -i})\\ 
&= & 2^{L-1}\delta_{2^{L-1}} + \displaystyle{\sum_{i=1}^{2^{L-1}-1}} (\delta_i + \delta_{2^L -i}) + \displaystyle{\sum_{i=1}^{2^{L-1}-1}}(i-1)(\delta_i + \delta_{2^L -i}) \\ 
& = & (2^{L-1}-1)\delta_{2^{L-1}}- \delta_0+\underbrace{\delta_0 + \delta_{2^{L-1}} + \displaystyle{\sum_{i=1}^{2^{L-1}-1}} (\delta_i + \delta_{2^L -i})}_{n} + \\
& & \displaystyle{\sum_{i=1}^{2^{L-1}-1}}(i-1)(\delta_i + \delta_{2^L -i}) \\ 
&= & n + (2^{L-1}-1)\delta_{2^{L-1}} - \delta_0 + \displaystyle{\sum_{i=1}^{2^{L-1}-1}}(i-1)(\delta_i + \delta_{2^L -i}).
\end{IEEEeqnarray*}

For the $i$-th coordinate of $\vect{c}$, for $i\in [1:n]$, it holds for $p=2$ \cite[Prop.~3.2]{HengYue_15} that

\begin{IEEEeqnarray*}{c}
w_{\textnormal{H}}(\gray(c_i))=\begin{cases}
0, & \mbox{if} \  c_i=0. \\ 
2^{L-1}, & \mbox{if} \ c_i=2^{L-1}. \\ 
2^{L-2}, & \mbox{otherwise}. 
\end{cases}
\end{IEEEeqnarray*}

Therefore, 
\begin{IEEEeqnarray*}{rCl}
w_{\textnormal{H}}(\gray(\vect{c}))& =\displaystyle{\sum_{i=1}^{n}}w_{\textnormal{H}}(\gray(c_i))=2^{L-1}\delta_{2^{L-1}} + 2^{L-2}\displaystyle{\sum_{i=1}^{2^{L-1}-1}}\delta_i= 2^{L-2}\left(2\delta_{2^{L-1}} + \displaystyle{\sum_{i=1}^{2^{L-1}-1}}\delta_i\right) \\ &=2^{L-2}\left(\delta_{2^{L-1}} + \underbrace{\delta_{2^{L-1}} + \displaystyle{\sum_{i=1}^{2^{L-1}-1}}\delta_i + \delta_0}_{n} - \delta_0 \right) =2^{L-2}(n + \delta_{2^{L-1}}- \delta_0). \hfill \qedhere
\end{IEEEeqnarray*}
\end{proof}

\begin{corollary} \label{cor:distpreserved}
Consider $\vect{c}=(c_1, c_2, \cdots, c_n) \in \code{C} \subseteq \Integers_{2^L}^n$, then $w_{\textnormal{H}}(\gray(\vect{c}))=\Lwt{\vect{c}}$ if, and only if, 
\begin{IEEEeqnarray*}{c}
\displaystyle{\sum_{i=1}^{2^{L-1}-1}}(i-1)(\delta_i + \delta_{2^L -i})=(2^{L-2}-1)(n - \delta_{2^{L-1}}- \delta_0).
\end{IEEEeqnarray*}
\end{corollary}

\begin{proof}
    It is straightforward, since
    \begin{IEEEeqnarray*}{c}
    w_{\textnormal{H}}(\gray(\vect{c}))-\Lwt{\vect{c}}=(2^{L-2}-1)(n - \delta_{2^{L-1}}- \delta_0) -\displaystyle{\sum_{i=1}^{2^{L-1}-1}}(i-1)(\delta_i + \delta_{2^L -i}).  \hspace{0.6cm} \qedhere
    \end{IEEEeqnarray*}
\end{proof}

\begin{example}
If $L=2$, then $w_{\textnormal{H}}(\gray(\vect{c}))=\Lwt{\vect{c}}$, for all $\vect{c} \in \code{C} \subseteq \Integers_{2^L}^n$. For $L=3$, then $w_{\textnormal{H}}(\gray(\vect{c}))=\Lwt{\vect{c}}$, for all $\vect{c} \in \code{C} \subseteq \Integers_{2^L}^n$ such that $\delta_1 + \delta_7=\delta_3 + \delta_5$.
\exampleend
\end{example}

Let $\code{C} \subseteq \Integers_{2^L}^n$ a $\Integers_{2^L}$-additive code. From Proposition~\ref{prop:lee-hammimng-weight} we can conclude that the codeword $\vect{c} \in \code{C}$ of minimum Lee weight is such that $(2^{L-1}-1)\delta_{2^{L-1}} + \displaystyle{\sum_{i=1}^{2^{L-1}-1}}(i-1)(\delta_i + \delta_{2^L -i}) - \delta_0$ assumes the smallest possible value. Moreover, for the Hamming weight to be minimum, $\vect{c} \in \code{C}$ should be such that $\delta_{2^{L-1}}- \delta_0$ assumes the smallest possible value. Observe that, in general, a minimum weight vector $\vect{c} \in \code{C}$ does not necessarily induce a minimum weight vector in $\gray(\code{C})$, as we present next.

\begin{example}
\label{ex:distance-comparison}
Given the code $\code{C}=\langle (2,0,1,4),(0,2,3,6),(1,0,2,4)\rangle \subset \Integers_{8}^4$, we have that $\vect{c}=(1,0,7,0) \in \code{C}$ is a minimum Lee weight codeword such that $\Lwt{\vect{c}}=2=d_{\textnormal{Lee}}(\code{C})$, but 
\begin{IEEEeqnarray*}{c}
w_{\textnormal{H}}(\gray(\vect{c}))=w_{\textnormal{H}}((0,1,0,1,0,0,0,0,1,0,0,1,0,0,0,0))=4 \neq 2=d_{\textnormal{H}}(\gray(\code{C})).
\end{IEEEeqnarray*}
However, for $\code{C}=\langle (1,7,4),(0,1,1),(0,0,5)\rangle \subset \Integers_{8}^3$, we have that $\vect{c}=(0,0,2) \in \code{C}$ is such that $\gray(\vect{c})$ achieves the minimum Hamming distance, i.e., $w_{\textnormal{H}}(\gray(\vect{c}))=w_{\textnormal{H}}((0,0,0,0,0,0,0,0,0,0,1,1))=d_{\textnormal{H}}(\gray(\code{C}))=2$, but $\Lwt{\vect{c}}=2 \neq 1=d_{\textnormal{Lee}}(\code{C})$.
\exampleend
\end{example}

We end this section with a classification of some $\Integers_{2^L}$-additive codes in Table~\ref{tab:table-gray-linearity}. 
The linearity~(L)~/~nonlinearity~(NL) of the corresponding binary codes $\gray(\code{C})$ was verified with the software \emph{Wolfram Mathematica}~\cite{Mathematica}, based on the conditions stated in Theorem~\ref{thm:iff_condition_concatenated} and Corollary~\ref{coro:princ}
Also, we tabulate the minimum Lee distance of the $\Integers_{2^L}$-additive code $\code{C}$ and the minimum Hamming distance of its respective $\Integers_{2^L}$-linear code. For $L=2$, the codes listed were selected through complete search, and are optimum with respect to the Lee distance and the number of codewords of minimum Lee weight (the kissing number), for each given length $n$. The searches were carried out, respectively, without restrictions on the generator matrices (NL), and restricted to generator matrices that meet the condition of Lemma~\ref{lem:simple-condition} (L).

\begin{sidewaystable}[htb!]
\caption{Examples of $\Integers_{2^L}$-linear codes from $\Integers_{2^L}$-additive codes}
\label{tab:table-gray-linearity}
\centering
\begin{tabular}{|c|c|c|c|c|c|c|c|}%
\hline
$L$ & $n$ & $\code{C} \subseteq \Integers_{2^L}^n$ & $d_{\textnormal{Lee}}(\code{C})$  & $\code{C}_i \circ \code{C}_i \subseteq \code{C}_{i+1} $ & $\mathcal{R} \subset \code{B}?$ &  $\gray(\code{C})$ &      $d_{\textnormal{H}}(\gray(\code{C}))$ \\
 &  &  &   & $\forall i$ (Cor.~\ref{coro:princ})   & (Th.~\ref{thm:iff_condition_concatenated})&  &  
\\ \hline
$2$ & $6$ & $\langle ( 1, 0, 0, 0, 1, 2),
( 0, 1, 0, 1, 0, 2),
( 0, 0, 1, 2, 2, 1)\rangle$   & $4$ & \cmark  &  \cmark  & L & $4$ \\
$2$ & $6$ & $\langle ( 1, 0, 0, 1, 1, 1),
( 0, 1, 0, 1, 2, 3),
( 0, 0, 1, 1, 3, 2)\rangle$   & $4$ & \xmark & \xmark  & NL & $4$ \\
$2$ & $8$ & $\code{O}$~(Ex.~\ref{ex:nordstrom-robinson}) & $6$ & \cmark & \xmark  & NL & $6$  \\
$2$ & $8$ & \vspace{-0.1cm} $\langle (1, 0, 0, 0, 0, 1, 2, 2), (0, 1, 0, 0, 1, 2, 0, 2),$ 
& $4$ & \cmark & \cmark & L & $4$ \\ 
& & $(0, 0, 1, 0, 2, 0, 1, 2), (0, 0, 0, 1, 2, 2, 2, 1)\rangle$ & & & & & \\
$2$ & $10$ & \vspace{-0.1cm} $\langle (1, 0, 0, 0, 0, 1, 1, 3, 3, 3), (0, 1, 0, 0, 0, 1, 2, 0, 1, 1),$
& $6$ &\xmark & \xmark  & NL & $6$  \\
&  & \vspace{-0.1cm} $(0, 0, 1, 0, 1, 0, 1, 1, 1, 1), (0, 0, 0, 1, 1, 0, 1, 2, 0, 3),$ &  & &   &  &   \\
&  &  $(0, 0, 0, 0, 2, 0, 0, 0, 2, 2), (0, 0, 0, 0, 0, 2, 0, 2, 2, 0)\rangle$ &  & &   &  &   \\
$2$ & $10$ & \vspace{-0.1cm} $\langle (1, 0, 0, 0, 0, 0, 0, 1, 2, 2), (  0, 1, 0, 0, 0, 0, 1, 0, 1, 2),$ 
& $4$ & \cmark & \cmark & L & $4$ \\ 
&  & \vspace{-0.1cm} $(0, 0, 1, 0, 0, 0, 2, 2, 0, 1), ( 0, 0, 0, 1, 1, 1, 0, 0, 0, 2),$ &  & &   &  &   \\
&  & \vspace{-0.1cm}  $(0, 0, 0, 0, 2, 0, 0, 2, 0, 2),(0, 0, 0, 0, 0, 2, 2, 0, 2, 2)\rangle$ &  & &   &  &   \\
$3$ & $3$ & $\langle (3,5,7) \rangle$  & $5$ & \cmark & \xmark & NL  & $6$  \\
$3$ & $3$ & $\langle (2,1,5), (0,3,6), (0,0,7) \rangle$  & $1$ & \cmark & \cmark  & L & $2$  \\
$3$ & $3$ & $\langle (0,2,4),(2,4,6) \rangle$ & $4$ &\cmark & \cmark  & L & $4$ \\
$3$ & $4$ & $\langle (2,0,1,4), (0,2,3,6), (1,0,2,4) \rangle$~(Ex.~\ref{ex:distance-comparison})  & $2$  & \cmark & \cmark  & L  & $2$ \\
$3$ & $8$ & $\langle (0,1,2,3,4,5,6,7)\rangle$~(Ex.~\ref{ex:tela-preta-gustavo}) & $16$ & \xmark & \xmark  & NL  & $16$ \\
$3$ & $8$ & $\langle (4,0,2,3,0,4,6,2)\rangle$ & $4$ & \cmark & \cmark & L  & $4$  \\
$3$ & $8$ & $\langle (0,1,2,3,4,5,6,7),(1,1,1,1,1,1,1,1) \rangle$ & $8$ & \cmark & \cmark  & L & $16$  \\
$4$ & $5$ & $\langle (8,2,0,10,8),(0,6,8,14,15) \rangle$ & $2$ & \cmark &  \cmark & L & $4$ \\
$4$ & $5$ & $\langle (4,5,8,9,12), (0,6,8,14,15) \rangle$ & $2$ & \cmark &  \xmark & NL & $4$ \\
$4$ & $5$ & $\langle (0,6,8,14,15) \rangle$ & $8$ & \cmark &  \cmark & L & $8$ \\
\hline
\end{tabular}%
\end{sidewaystable}

\section{Nested Construction of Linear $\Integers_{2^L}$-Linear Codes}
\label{sec:linear-gray-codes}

Consider a $\Integers_{2^L}$-additive code $\code{C}$, such that its decomposition code $\code{B}$ is simply the cartesian product of the associated codes, i.e., $\code{B}=\code{C}_1\times \code{C}_2 \times \dots \times \code{C}_L$. Hence, $\code{C} =\code{C}_1+2\code{C}_2+\dots+2^{L-1}\code{C}_L$ and we can directly obtain linear $\Integers_{2^L}$-linear codes by the so called \emph{nested construction} we introduce next.

    One can apply~\cite[Th.~1]{KositwattanarerkOggier_13} to conclude that $\code{C} =\code{C}_1+2\code{C}_2+\dots+2^{L-1}\code{C}_L$ is $\Integers_{2^L}$-additive if and only if the associated codes $\code{C}_1, \code{C}_2, \dots, \code{C}_L$ are closed under Schur product, i.e., $\code{C}_i \circ \code{C}_i \subseteq \code{C}_{i+1}$ for$i \in [1: L-1]$. This result suffices to claim the linearity of the $\Integers_{2^L}$-linear code.

\begin{theorem}[Nested construction]
\label{thm:concatenated-cartesian-product}
Consider $\code{C}_1, \code{C}_2, \dots, \code{C}_L \subseteq \Integers_2^n$ to be $L$ binary linear codes. The \emph{nested}  code $\code{C}=\code{C}_1+2\code{C}_2+\dots+2^{L-1}\code{C}_L$ is $\Integers_{2^L}$-additive if and only if $\gray(\code{C})$ is linear.
\end{theorem} 

\begin{proof} The necessary condition follows directly from Corollary~\ref{coro:princ}. On the other hand, given that $\code{C}=\code{C}_1+2\code{C}_2+\dots+2^{L-1}\code{C}_L$ is $\Integers_{2^L}$-additive, we have that $\code{C}_i \circ \code{C}_i \subseteq \code{C}_{i+1}$ for $i \in [1: L-1]$~\cite[Th.~1]{KositwattanarerkOggier_13}. We want to show that $\gray(\code{C})$ is linear. By hypothesis and given that $\vect{0} \in \code{C}_1$ by linearity, we have a vector of the form
\begin{IEEEeqnarray*}{c}
(\vect{0},\vect{u}_1 \circ \vect{v}_1, \vect{u}_2 \circ \vect{v}_2, \dots, \vect{v}_{L-1} \circ \vect{v}_{L-1}) \in \code{C}_1 \times \dots \times \code{C}_L = \code{B},
\end{IEEEeqnarray*}
for all $\vect{c} = \sum_{i=1}^L 2^{i-1}\vect{u}_i, \vect{d} = \sum_{i=1}^L 2^{i-1}\vect{v}_i \in \code{C}$. Therefore, from Theorem~\ref{thm:iff_condition_concatenated}, we can conclude that $\gray(\code{C})$ is linear.
\end{proof}

Theorem~\ref{thm:concatenated-cartesian-product} provides an efficient way of constructing linear $\Integers_{2^L}$-linear codes. We start with $L$ binary linear codes $\code{C}_1, \code{C}_2, \dots, \code{C}_L$ that are closed under Schur product, construct the $\Integers_{2^L}$-additive code $\code{C}=\code{C}_1+2\code{C}_2+\dots+2^{L-1}\code{C}_L$, and we have a guarantee that $\gray(\code{C})$ is a linear $\Integers_{2^L}$-linear code.

As a consequence of the condition $\code{C}_i \circ \code{C}_i \subseteq \code{C}_{i+1}$, for $i \in [1: L-1]$, we have that $\code{C}_1 \subseteq \code{C}_2 \subseteq \dots \subseteq \code{C}_{L}$ which is why we chose to call it a \textit{nested construction}. We remark that by using such construction, we obtain a linear $\Integers_{2^L}$-linear code, and the intrinsic operation cost is to find binary linear codes that are closed under Schur product. Secs.~\ref{sec:linear-gray-codes-rm} and~\ref{sec:linear-gray-cyclic} consist of an application of the nested construction to two families of codes that have such property: Reed-Muller and cyclic codes.

\subsection{Reed-Muller Codes}
\label{sec:linear-gray-codes-rm}

We recall the class of binary Reed-Muller codes, according to \cite[Ch.~13]{MacWilliamsSloane77_1}.

\begin{definition}[Reed-Muller codes]
  \label{def:reed-muller-codes}
  For a given $m\in\Naturals$, the $r$-th order binary Reed-Muller code $\code{R}(r,m)$ is a 
  linear $[n=2^m,k=\sum_{i=0}^r\binom{m}{i}]$ code 
  for $0 \leq r \leq m$, constructed as the vector space spanned by the set of all $m$-variable Boolean monomials of degree at most $r$.
\end{definition}

An alternative construction of Reed-Muller codes is the $(\vect{u},\vect{u}\oplusq{2}\vect{v})$-construction.

\begin{definition}[Reed-Muller codes via $(\vect{u}, \vect{u} \oplusq{2}\vect{v})$]
\begin{IEEEeqnarray*}{c}
\code{R}(r, m) = \{(\vect{u}, \vect{u} \oplusq{2} \vect{v}): \vect{u} \in \code{R}(r,m-1),~ \vect{v} \in \code{R}(r-1,m-1)\}.
\end{IEEEeqnarray*}
\end{definition}

    The following result states a property that will be useful in the construction of binary linear codes from $\Integers_{2^L}$-additive codes.

\begin{theorem}\cite[Th.~1.10.1]{HuffmanPless03_1}
\label{thm:nested-reed-muller}
If $0 \leq i \leq j \leq m$, then $\code{R}(i,m) \subset \code{R}(j,m)$.    
\end{theorem} 

    We are interested in constructing $\Integers_{2^L}$-additive codes from Reed-Muller codes. Hence, given $L$ Reed-Muller codes of the form $\code{R}(r_i,m)$, with $1 \leq i \leq L$, then the code
    \begin{IEEEeqnarray}{C}
    \code{C}_{\textnormal{RM}} = \code{R}(r_1, m) + 2\code{R}(r_2, m) + \dots + 2^{L-1}\code{R}(r_L,m) \subseteq \Integers_{2^L}^{2^m},
    \end{IEEEeqnarray}
    with $r_1 \leq r_2 \leq \dots \leq r_L$ is $\Integers_{2^L}$-additive if and only if the codes $\code{R}(r_1, m), \code{R}(r_2, m), \dots, \code{R}(r_L,m)$ are closed under Schur product~\cite[Th.~1]{KositwattanarerkOggier_13}.

General conditions on the Reed-Muller codes being closed under Schur product are presented next.

\begin{proposition}
\label{prop:chain-shur-sequence}
The Reed-Muller codes $\code{R}(0,m), \code{R}(1,m), \code{R}(2,m)$ are closed under Schur product for all $m$.  
\end{proposition}
\begin{proof} We aim to demonstrate that $\code{R}(i,m) \circ \code{R}(i,m)  \subset  \code{R}(i+1,m)$, for $i \in \{0,1\}$ and all $m$. For $i=0$, observe that $\code{R}(0, m) \circ \code{R}(0, m)= \code{R}(0, m) \subset \code{R}(1, m)$ by Theorem~\ref{thm:nested-reed-muller}. For $i=1$, by definition, $\code{R}(1,m)$ is the set of vectors $f$, where $f(v_1^\ast, v_2^\ast, \dots, v_m^\ast)$ is a polynomial of degree at most $1$. Given that the Schur product of Boolean functions translates into the product of polynomials subject to the relation ${v_i^\ast}^2=v_i^\ast$, for all $i$, the product of polynomials in $\code{R}(1,m)$ will have degree at most $2$ and thus be an element of $\code{R}(2,m)$, as we wanted to demonstrate.
\end{proof}

An alternative proof of Proposition~\ref{prop:chain-shur-sequence} comes from the fact that for Reed-Muller codes, it holds for $m \geq 2$ that $\code{R}(r,m)\circ\code{R}(r',m)=\code{R}(r+r',m)$~\cite[p.~4]{Randriambololona15}.

From the construction of Barnes-Wall lattices from Reed-Muller codes~\cite{HuNebe_20} we recall the following additional result on Reed-Muller codes.
\begin{proposition}\cite[p.~12]{HuNebe_20}
\label{prop:rm-chain-schur}
Consider the sets of Reed-Muller codes
\begin{IEEEeqnarray*}{rCl}    
\code{R}_{\textnormal{e}}:& & ~\code{R}(0,2\ell), \code{R}(2,2\ell), \dots, \code{R}(2\ell-2,2\ell) \\
\code{R}_{\textnormal{o}}:& & ~\code{R}(1,2\ell), \code{R}(3,2\ell), \dots, \code{R}(2\ell-1,2\ell).
\end{IEEEeqnarray*}
For $\ell<4$, $\code{R}_{\textnormal{e}}$ and $\code{R}_{\textnormal{o}}$ are closed under Schur product.
\end{proposition}

    Consequently, Propositions~\ref{prop:chain-shur-sequence} and~\ref{prop:rm-chain-schur} provide the necessary tools to construct linear $\Integers_{2^L}$-linear codes via the nested construction, or Theorem~\ref{thm:concatenated-cartesian-product}.

\begin{corollary}
Consider the $\Integers_8$-additive code $\code{C}_{\textnormal{RM}} = \code{R}(0,m) + 2\code{R}(1,m) + \code{R}(2,m)$ of length $2^m$.  Then, the $\Integers_{8}$-linear code $\gray(\code{C}_{\textnormal{RM}})$ is linear.
\end{corollary} 


\begin{corollary} Consider the $\Integers_{2^\ell}$-additive codes $\code{C}_{\textnormal{RM}} = \code{R}(0,2\ell) + 2\code{R}(2,2\ell) + \dots + 2^{\ell-1}\code{R}(2\ell-2,2\ell)$ and $\code{C}_{\textnormal{RM}}' = \code{R}(1,2\ell) + 2\code{R}(3,2\ell) + \dots + 2^{\ell-1}\code{R}(2\ell-1,2\ell)$ of length $2^{2\ell}$ with $\ell <4$, obtained from $\code{R}_\textnormal{e}$ and $\code{R}_{\textnormal{o}}$ as in Proposition~\ref{prop:rm-chain-schur}.  Then, the respective $\Integers_{2^\ell}$-linear codes $\gray(\code{C}_{\textnormal{RM}})$ and $\gray(\code{C}_{\textnormal{RM}}')$ are linear.
\end{corollary}

\subsection{Cyclic Codes}
\label{sec:linear-gray-cyclic}

A complete classification of $\Integers_{4}$-linear codes of odd length from cyclic $\Integers_4$-additive codes was done in~\cite{Wolfmann01}. We expand on that and apply the nested construction to obtain cyclic $\Integers_{2^L}$-additive codes $\code{C}_{\text{cyc}}$ such that $\gray(\code{C}_{\text{cyc}})$ is a linear $\Integers_{2^L}$-linear code. We also explore characteristics of the squares of cyclic codes in the context of the nested constructions.

\begin{definition}[Cyclic code] A $\Integers_{2^L}$-additive code $\code{C}$ is 
cyclic if for any codeword $(c_1,\dots, c_{n-1},$ $c_n) \in \code{C}$, the cyclic shift $(c_n,c_1,\dots, c_{n-1}) \in \code{C}$.
\end{definition}

To simplify notation, a cyclic $\Integers_{2^L}$-additive code will be 
denoted as $\Integers_{2^L}$-cyclic code from now on, i.e., the linearity will be implicit.

Observe that each codeword can be associated with a polynomial as
\begin{IEEEeqnarray*}{rCl}
\code{C} \subseteq \Integers_{2^L}^n & \rightarrow & R_{2^L}[x]=\frac{\Integers_{2^L}[x]}{\langle x^n-1 \rangle} \\
(c_0, \dots, c_{n-2},c_{n-1}) & \mapsto & c(x)=c_0 + \dots + c_{n-2}x^{n-2} + c_{n-1}x^{n-1}.
\end{IEEEeqnarray*}
Multiplications of the form $x^k c(x) \bmod (x^n-1)$, where $k\in [1: n-1]$ are cyclic shifts of $c(x)$. Therefore, a cyclic code can also be identified as an ideal of $ R_{2^L}[x]$. 

\begin{proposition} \label{prop:cyclic-code-associated} If $\code{C}$ is a $\Integers_{2^L}$-cyclic code then its associated codes $\code{C}_1,\code{C}_2, \dots, \code{C}_L \subseteq {\Integers_2^n}$ are cyclic.
\end{proposition}

\begin{proof} Given any $\vect{c} \in \code{C} \subset \Integers_{2^L}^n$, we recall from~\eqref{latticetoring} that $\vect{c} = \sum_{i=1}^L 2^{i-1}\vect{u}_i$, where $\vect{u}_i = \left(u_{i,1},u_{i,2},...,u_{i,n} \right) \in \code{C}_i \subseteq \Integers_2^n$, where $\code{C}_i$ is the associated code for $i \in [1: L]$.
Since $\code{C}$ is cyclic, all the cyclic shifts of $\vect{c}$ belong to $\code{C}$, thus by definition, all cyclic shifts of $\vect{u}_i$ belong to the associated code $\code{C}_i$, for $i \in [1: L]$. 
\end{proof}

The reverse implication of Proposition~\ref{prop:cyclic-code-associated} holds with an analogous proof when the $\Integers_{2^L}$-additive code is generated as $\code{C}_{\text{cyc}}\eqdef\code{C}_1 + 2\code{C}_2 + \dots +  2^{L-1}\code{C}_L$, from $L$ binary codes. 

\begin{proposition}
\label{prop:cyclicity-special-case}
Consider $L$ binary cyclic codes $\code{C}_1, \dots, \code{C}_L \subseteq \Integers_{2}^n$ closed under Schur product. Then, the $\Integers_{2^L}$-additive code $\code{C}_{\text{cyc}}=\code{C}_1 + 2\code{C}_2 + \dots +  2^{L-1}\code{C}_L \subseteq \Integers_{2^L}^n$ is cyclic.
\end{proposition}

Proposition~\ref{prop:cyclicity-special-case} provides a way of generating $\Integers_{2^L}$-cyclic codes, given that the binary codes considered in the nested construction are closed under the Schur product. This condition can be addressed through the generator polynomials of the respective binary cyclic codes, as we proceed to do next. 

We start by recalling some additional definitions concerning cyclic codes over $\Integers_2$, which actually hold for finite fields in general, but we focus on $\Integers_2$ here. The aim is to present auxiliary results about the squares of cyclic codes from~\cite{Cascudo}.

Given $a \in \mathbb{Z}_n$, one defines the $q$-cyclotomic coset $\mathbb{C}_a := \{a, aq, ..., aq^{m_a -1}\} \subset \mathbb{Z}_n$, where $m_a$ is the smallest positive integer such that $aq^{m_q} \equiv a \mod n.$ Since $q$-cyclotomic cosets are, in particular, equivalence classes, then $\mathbb{Z}_n$ is partitioned into such $q$-cyclotomic cosets and $\mathbb{C}_a \cap \mathbb{C}_b =\emptyset$, if and only if, $\mathbb{C}_a \neq \mathbb{C}_b$. 

Let $\beta$ be a $n$-primitive root of unit in a finite extension field (algebraic closure) of $\Integers_2$. The minimal polynomial of $\beta^a$ is defined as $p^{a}(x) = \prod_{i \in \mathbb{C}_a} (x-\beta^i)$ and $x^n -1$ may be described as product of all minimal polynomials, namely, minimal polynomials are described based on individual $q$-cyclotomic cosets and, in general, generator polynomials of cyclic codes are described based on union of some $q$-cyclotomic cosets.

Let $\code{C}_g =\langle g(x)\rangle \subset R_{2}[x]$ be a binary cyclic code, namely, $g(x)$ is a divisor of $x^n -1 \in \Integers_2 [x].$ The \emph{defining set of $\code{C}_g$} is the set $J:=\{j\in \mathbb{Z}_n : g(\beta^j ) = 0\}$. In particular, we define $I\eqdef\Integers_n \setminus J$ and $-I \eqdef \{-i: i \in I\} \subseteq \Integers_n$.

\begin{definition}[Sum of subsets]
For subsets $A,B \subseteq \Integers_n,$ define $A+B :=\{i + j: i\in A, j \in B\}\subseteq \Integers_n$.
\end{definition}

Observe that the sets $I_1 , I_2 , \ldots, I_L$ are related to the associated codes $\code{C}_1, \code{C}_2, \ldots, \code{C}_L$ of $\code{C} \subseteq \mathbb{Z}_{2^L} ^n$, which play an important role when describing the chain of the squares cyclic codes that satisfy the Schur product condition. Moreover, we denote $I + I = 2I$, $I+ I+I+I = 2I + 2I = 4I$, and so on. 

\begin{definition}[Evaluation sets] For a set $M\subseteq \{1,\dots, n-1\}$, let $\mathcal{P}(M) \eqdef \left\{ \sum_{i \in M} f_iX^i: f_i \in \Integers_{2^r} \right\}$. In addition, define the vector space $\mathcal{B}(M) \eqdef \{\left(f(1),f(\beta),\dots, f(\beta^{n-1})\right): f \in \mathcal{P}(M)\}$.
\end{definition}

\begin{lemma}\cite[Lemma~5]{Cascudo} Let $\code{C}$ be the cyclic code generated by $g(x) = \nicefrac{(x^n - 1)}{\prod_{i\in I} \left(x - \beta^i \right)} \in \Integers_2[x]$. Then, $\code{C}=\mathcal{B}(-I)|_{\Integers_2}$.
\end{lemma}

    Therefore, the generator polynomial of the square of a cyclic code is known.

\begin{theorem}\cite[Th.~1]{Cascudo}\label{maincascudosth}
 If $\code{C} = \mathcal{B} (-I)|_{\Integers_2}$, then $\code{C}^2 = \mathcal{B} (-(I + I))|_{\Integers_2}$. In other words, if $\code{C}$ is a cyclic code generated by the polynomial $g(x) = \nicefrac{(x^n - 1)}{\prod_{i\in I} \left(x - \beta^i \right)}$ then $\code{C}^2$ is a cyclic code with the generator polynomial $g'(x) = \nicefrac{(x^n - 1)}{\prod_{\ell \in I+ I} \left(x - \beta^{\ell} \right)}$.
\end{theorem}

Therefore, we are now ready to construct a cyclic $\Integers_{2^L}$-additive code $\code{C} \subseteq \Integers_{2^L}^n$ such that its respective $\Integers_{2^L}$-linear code is linear. We proceed as follows:
\begin{enumerate}
    \item[1.] Start with a binary cyclic code $\code{C}_1 \subseteq \Integers_2^n$ and calculate $\code{C}_1^2$. 
    \item[2.] For each level $k = 2,\dots, L$, let $\code{C}_k = \code{C}_{k-1}^2$. Notice that each $\code{C}_k$ is also cyclic due to Theorem~\ref{maincascudosth}.
    \item[3.] Define the $\Integers_{2^L}$-additive code as $\code{C}_{\text{cyc}} \eqdef \code{C}_1+2\code{C}_2+\dots+2^{L-1}\code{C}_L$, which is cyclic according to Proposition~\ref{prop:cyclicity-special-case}. 
\end{enumerate}

Theorem~\ref{thm:concatenated-cartesian-product} then guarantees that $\gray(\code{C}_{\text{cyc}})$ is a linear $\Integers_{2^L}$-linear code.

\begin{example}
\label{ex:cyclic-nested-construction}
Let $\code{C}_1$ be a cyclic $[7,3,4]$-code in $\frac{\Integers_2 [x]}{\left\langle x^7 -1 \right\rangle}$, where $\code{C}_1 = \langle x^4 +x^2 +x +1\rangle = \linebreak \langle (x+1)(x^3 + x^2 +1)\rangle$. In addition, let $\beta$ be a $7$-primitive root of unity, so $J_1 = \{0,3,5,6\}$ is the defining set of $\code{C}_1$ and $I_1 =\Integers_{7} \setminus J_1 =\{1,2,4\}$. Then,
$x^4 +x^2 +x +1 = \tfrac{x^n -1}{ \prod_{\ell \in I_1} (x -\beta^{\ell})}$.    

Observe that $I_1+I_1=I_2=\{1,2,3,4,5,6\}$ and $J_2=\{0\}$. From Theorem~\ref{maincascudosth}, the cyclic code $\code{C}_2=\code{C}_1^2$ is a $[7,6,2]$-code generated by
$x+1 =\tfrac{x^7 -1 }{\prod_{\ell \in I_1 + I_1} (x - \beta^{\ell})} = \tfrac{x^7 -1}{\prod_{\ell \in I_2} (x - \beta^{\ell})}$. 
For the next step, $I_2 + I_2 = \mathbb{Z}_7$ and the generator polynomial of $\code{C}_3=\code{C}_2^2=\Integers_2^7$ is
$1=\tfrac{x^7 -1 }{\prod_{\ell \in I_2 + I_2} (x - \beta^{\ell})}= \tfrac{x^7 -1 }{\prod_{\ell \in \mathbb{Z}_7} (x - \beta^{\ell})}$.
    
Hence, according to Theorem~\ref{maincascudosth}, $\code{C}_1 \subset \code{C}^2 _1  = \code{C}_2 \subset \code{C}^2 _2 = \code{C}_3 =\Integers_2 ^7$ and from Theorem~\ref{thm:concatenated-cartesian-product}, we know 
that $\gray(\code{C}_{\text{cyc}})$ is linear, for $\code{C}_{\text{cyc}}=\code{C}_1+2\code{C}_2+4\code{C}_3 \subseteq \Integers_8^7$ cyclic. 
\hfill\exampleend
\end{example}

Next, we show that if $\code{C}$ is a $\Integers_{2^L}$-cyclic code and $\gray(\code{C})$ is linear, then $\gray(\code{C})$ is quasi-cyclic, which in particular applies to $\code{C}_{\text{cyc}}$. 
A linear code $\code{A}$ is said to be \textit{quasi-cyclic} of index $s$ if the shift of a codeword by $s$ positions is a codeword of $\code{A}$. Therefore, cyclic codes are quasi-cyclic for $s=1$.

\begin{proposition}\label{imageqccode} Let $\code{C} \subseteq \Integers_{2^L}^n$ be a cyclic code. Then $T^{2^{L-1}} (\gray(\code{C})) = \gray(\code{C})$, where $T$ represents the cyclic-shift operator and $T^{2^{L - 1}} = T^{2^{L-2}}\circ T$. In particular, if $\gray(\code{C})$ is linear, then it is a quasi-cyclic code of index $2^{L-1}.$ 
\end{proposition}

\begin{proof} Consider $\vect{c}=(c_1, \dots, c_{n-1}, c_n) \in \code{C}$ and $\vect{c}' = (c_n, c_1, \dots, c_{n-1})$ be its cyclic shift. Since $\code{C}$ is cyclic, $\vect{c}' \in \code{C}$. Then, by definition, we have that $\gray(\vect{c}) =  (\gray(c_1),\dots,\gray(c_{n-1}),\gray(c_n))$, where $\gray(c_j) = (u_{j,L},\dots, u_{j,L}) + (u_{j,1},u_{j,2},\dots, u_{j,L-1})\mat{Y}$. Observe that $\gray(\vect{c}')=(\gray(c_n), \gray(c_1), \dots,$ $\gray(c_{n-1}))$, which is the cyclic shift of $\gray(\vect{c})$ by $2^{L-1}$ positions.
\end{proof}

Proposition~\ref{imageqccode} extends~\cite[Th.~3.9]{Wolfman}. Indeed, let $R$ be a left-shift operator, which is naturally a permutation in $S_{2^{L-1}n}$\footnote{Group of permutations of $2^{L-1}n$ elements.}. Hence, if $\code{C}$ is a cyclic code and $\gray(\code{C})$ is linear, $R^{2^{L-1}-1}\left(\gray(\code{C})\right)$ is a cyclic code in $\mathbb{F}_2 ^{2^{L-1}n}$, once $T^{2^{L-1}}(\gray(\code{C}))=\gray(\code{C})$. In other words, $\gray(\code{C})$ is a cyclic code up to permutations. Observe that the $\Integers_8$-linear code from Example~\ref{ex:cyclic-nested-construction} is linear and quasi-cyclic.

In the construction of $\code{C}_{\text{cyc}}$, one can notice that if we reach a level $k$ such that $\code{C}_k=\code{C}_k^2$, or more than that, if $\code{C}_k=\code{C}_k^2=\Integers_2^n$, there is a trivial way of proceeding. The latter reflects on the number of level codes $L$, i.e., once we reach the universe code by squaring a cyclic code, that level dictates the cardinality of the alphabet $\mathbb{Z}_{2^L}$ where $\code{C}_{\text{cyc}}$ is defined. Thus, to have linear $\Integers_{2^L}$-linear with larger $L$, one needs to avoid this behavior. From this point forward, we will concentrate on that study.

For instance, the cyclic code $\code{C}_1$ considered in Example~\ref{ex:cyclic-nested-construction} is a simplex code. If $\code{C}$ is a $[n,k]$-simplex code, then it is known that $\code{C}^2$ is has dimension $\nicefrac{k(k+1)}{2}$~\cite[Ex.~2.2.5]{Mirandola12}. Moreover, in~\cite{CascudoCramerMirandolaZemor_15} the authors show that given a random $[n,k]$-code $\code{C}$,  then with high probability $\code{C}^2$ has dimension $\min\left\{n,\nicefrac{k(k+1)}{2} \right\}$. In order to study more general cases, we proceed with the following definition.

\begin{definition}
A chain of codes $\code{C}_1,\dots,\code{C}_L$ 
stabilizes at a level $t$, if $\code{C}_t ^2 = \code{C}_{t+1} = \code{C}_{t+1} ^2$, for the smallest $t \in [1:L]$. 
\end{definition}

The goal is to determine the smallest $t$-th level that the chain stabilizes. We will address this question for cyclic codes with length $n=p^m$, where $p$ is an odd prime and $m$ is a positive integer, and we start with the following result from number theory.   

\begin{theorem}\cite[Th. 40]{Niven}\label{throot}
If $p$ is an odd prime and $g$ is a primitive root of unity modulo $p^2$, then $g$ is a primitive root modulo $p^{\alpha}$ for $\alpha=3,4,5, ...$.  
\end{theorem}

We focus now on cases where $2$ is a primitive root of unit modulo $p$ and modulo $p^2$. Consequently, based on Theorem~\ref{throot}, so is it modulo $p^m$, for any integer $m\geq 2$. 

\begin{proposition}
Let $2$ be a primitive root of unity modulo $p^m$. Then, $\mathbb{Z}_{p^m}$ may be partitioned on $2$-cyclotomic cosets $\mathbb{C}_0, \mathbb{C}_1 , \mathbb{C}_p , \mathbb{C}_{p^2} ,\ldots, \mathbb{C}_{p^{m-1}}$, where  
$|\mathbb{C}_{p^i}| = \varphi(p^{m-i})$, for any $i \in [0: m]$, where $\varphi(\cdot)$ denotes the Euler's Totient function\footnote{The Euler's totient function is defined as $\varphi(n)=n\prod_{p\mid n} \bigl( 1- \frac{1}{p}\bigr)$, where the product is considered over distinct primes that divide $n$.}.
\end{proposition}

\begin{proof}
Given that $2$ is a primitive root of unity modulo $p^m$, let $\mathbb{C}_{p^i}= \{p^i , p^i 2 , ..., p^i 2^{o_{p_i} -1}\}$. It is straightforward that $o_{p_i} = \varphi(p^{m-i})$, for any $i \in [0:m]$. Notice that $\mathbb{C}_{p^i} \cap \mathbb{C}_{p^j}=\emptyset$, for $i >j$. Indeed, without loss of generality for $r>s$,
\begin{IEEEeqnarray*}{rCl}
p^i 2^r \equiv p^j 2^s \mod p^m &\Rightarrow& p^i 2^{r-s} \equiv p^j \mod p^m \Rightarrow p^{m-i} p^i 2^{r-s} \equiv p^{m-i} p^j \mod p^m \nonumber \\
&\Rightarrow& p^{m-(i-j)}\equiv 0\mod p^m ,  
\end{IEEEeqnarray*}
an absurd. Finally, 
\begin{IEEEeqnarray*}{rCl}
|\mathbb{C}_{0}| + |\mathbb{C}_{1}|+ |\mathbb{C}_{p}| +  ...+|\mathbb{C}_{p^{m-1}}| &=& 1 + \varphi(p^m) +\varphi(p^{m-1}) + ...+\varphi(p)    \nonumber\\
&=& \sum_{d|p^m} \varphi(d) = p^m. \hspace{4.7cm} \qedhere 
\end{IEEEeqnarray*} 
\end{proof}

\begin{lemma}\label{cyclotcosetop}
Let $\mathbb{C}_0 ,\mathbb{C}_1 , \mathbb{C}_p ,\ldots , \mathbb{C}_{p^{m-1}}$ be the $2$-cyclotomic cosets which partition the ring $\mathbb{Z}_{p^m}$. Then,
\begin{enumerate}
\item $\mathbb{C}_{p^i} +\mathbb{C}_{p^i}=2\mathbb{C}_{p^i} =\{0\}\cup \mathbb{C}_{p^i} \cup \mathbb{C}_{p^{i+1}} \cup ...\cup \mathbb{C}_{p^{m-1}}$, for any $i \in [0: m-1]$.  
\item $\mathbb{C}_{p^i} +\mathbb{C}_{p^j} = \mathbb{C}_{p^i} $, for $0\leq i <j \leq m-1$.
\item $(\mathbb{C}_{p^i} \cup \mathbb{C}_{p^j}) +   (\mathbb{C}_{p^i} \cup \mathbb{C}_{p^j})= 2\mathbb{C}_{p^i} = \{0\}\cup \mathbb{C}_{p^i} \cup \mathbb{C}_{p^{i+1}} \cup ...\cup \mathbb{C}_{p^{m-1}}$, for $i <j.$ Furthermore, $(\mathbb{C}_{0} \cup \mathbb{C}_{p^i}) +   (\mathbb{C}_{0} \cup \mathbb{C}_{p^i}) = 2\mathbb{C}_{p^i} = \{0\}\cup \mathbb{C}_{p^i} \cup \mathbb{C}_{p^{i+1}} \cup ...\cup \mathbb{C}_{p^{m-1}}$, for any $i \in [0: m-1]$. 
\end{enumerate}
\end{lemma}

\begin{proof}
\begin{enumerate}
\item Given $ i < j \leq m -1$, we have $p^{j-i} -1 \not\in \mathbb{C}_{p^t}$, for any $1 \leq t \leq m-1$, namely, $p^{j-i} -1 \in \mathbb{C}_{1}$, i.e., there exists $0\leq s_j \leq r_j \leq \varphi(p^m) -1$ such that $2^{r_j - s_j}\equiv p^{j-i} -1 \mod p^m$. Thus,
\begin{IEEEeqnarray}{rCl}\label{CpiCpi}
2^{r_j - s_j }\equiv p^{j-i} -1 \mod p^m &\Rightarrow& 2^{r_j -s_j} +1 \equiv p^{j-i} \mod p^m \nonumber \\
&\Rightarrow& p^i(2^{r_j -s_j } +1) \equiv p^{j} \mod p^m  \nonumber \\ 
&\Rightarrow& p^i(2^{r_j - s_j } +1)2^{s_j} \equiv p^{j}2^{s_j}  \mod p^m \nonumber \\
&\Rightarrow&      p^i 2^{r_j} + p^i 2^{s_j} \equiv p^{j} 2^{s_j}  \mod p^m.
\end{IEEEeqnarray}
Running $r_j$ and $s_j$ from $0$ to $\varphi(p^{m-j})-1$, all sums of elements of $\mathbb{C}_{p^i}$ yield elements in $\mathbb{C}_{p^j}$, i.e., $ 2\mathbb{C}_{p^i} \subset \mathbb{C}_{p^j}$, for each $i<j\leq m-1$. In particular, assuming $i=j$ in~\eqref{CpiCpi}, we have $2\mathbb{C}_{p^i}  = \mathbb{C}_{p^i}$. Moreover, since $p^i , p^i (p^{m-i} -1) \in \mathbb{C}_{p^i},$ we also have $p^i + p^i (p^{m-i} -1) = p^i (p^{m-i}) = p^m $, i.e., $0\in 2\mathbb{C}_{p^i}$. Thus, it holds that $   2\mathbb{C}_{p^i} \subset \{0\}\cup \mathbb{C}_{p^i} \cup \mathbb{C}_{p^{i+1}} \cup ...\cup \mathbb{C}_{p^{m-1}} $, for any $i \in [0:m-1]$. 
On the other hand, let $2^k p^j \in \mathbb{C}_{p^j}$, where $i\leq j\leq m-1$. Then 
\begin{IEEEeqnarray*}{c}
2^k p^j = 2^{k-1}p^{j-i} 2 p^i =  2^{k-1}p^{j-i}( p^i + p^i) =  (2^{k-1}p^{j-i})p^i + (2^{k-1}p^{j-i}) p^i  ,
\end{IEEEeqnarray*}
and $\mathbb{C}_{p^j} \subset 2\mathbb{C}_{p^i}$. Thus, $0 \in 2\mathbb{C}_{p^i}$, i.e., $\mathbb{C}_0 \subset 2\mathbb{C}_{p^i}$ and the inverse inclusion holds, from where the result follows.

\item For $i< j$, let $p^i 2^s \in \mathbb{C}_{p^i}$ and  $p^j 2^r \in \mathbb{C}_{p^j}$. If $s\leq r$, $p^i 2^s + p^j 2^r = (1 + p^{j-i}2^{r-s})p^i 2^s\mod p^m$; otherwise $p^i 2^s + p^j 2^r = (2^{s-r} + p^{j-i})p^i 2^r \mod p^m.$ Since $1 + p^{j-i}2^{r-s} \not\equiv 0 \mod p$ and $(2^{s-r} + p^{j-i})\not\equiv 0 \mod p,$ once $p$ is a odd prime number.

\item For $i\leq j$, based on items $[(i)]$ and $[(ii)]$ we have $(\mathbb{C}_{p^i} \cup \mathbb{C}_{p^j}) +   (\mathbb{C}_{p^i} \cup \mathbb{C}_{p^j}) = (\mathbb{C}_{p^i} +\mathbb{C}_{p^i})+ (\mathbb{C}_{p^i} +\mathbb{C}_{p^j}) + (\mathbb{C}_{p^i} +\mathbb{C}_{p^j}) + (\mathbb{C}_{p^j} +\mathbb{C}_{p^j}) = (\mathbb{C}_{p^i} +\mathbb{C}_{p^i})+(\mathbb{C}_{p^i} +\mathbb{C}_{p^i})+   (\mathbb{C}_{p^j} +\mathbb{C}_{p^j})= (\mathbb{C}_{p^i} +\mathbb{C}_{p^i})+   (\mathbb{C}_{p^j} +\mathbb{C}_{p^j})= \mathbb{C}_{p^i} +\mathbb{C}_{p^i}. $ \hfill \qedhere 
%
\end{enumerate} 
\end{proof}

\begin{example}\label{Exsqcycliccodes}
Let $q=2$, $L=3$, $n=5^3 = 125$, and $\code{C}$ be a $\mathbb{Z}_8$-cyclic code, with associated cyclic codes $\code{C}_1 \subset \code{C}_2 \subset \code{C}_3$. Based on Theorem~\ref{maincascudosth}, we analyze the squares of cyclic codes closed under Schur product. 

Notice that $2$ is a primitive root of unit modulo $5$ and $25$; consequently, via Theorem~\ref{throot}, 2 is a primitive root of unit modulo $5^m$, for any $m$ positive integer. The $2$-cyclotomic cosets modulo 125 are $\mathbb{C}_0$ ,$\mathbb{C}_1$, $\mathbb{C}_5$, $\mathbb{C}_{25}$, and the associated minimal polynomials are, respectively
\begin{IEEEeqnarray*}{rCl}
\begin{array}{ll}
p^0 (x)=x+1,& p^1 (x) =x^{100} + x^{75} + x^{50} + x^{25}+1,\\ 
p^{5} (x) = x^{20}+x^{15}+x^{10}+x^5 +1,  & p^{25} (x)= x^4 + x^3 + x^2 + x+1.    
\end{array}
\end{IEEEeqnarray*}

Let $g_i (x)$ be the generator of the associated cyclic codes for $i=1,2,3$. Besides that, define $I_2 = 2I_1$ and $I_3 = 2I_2$. From Theorem~\ref{maincascudosth} and Lemma~\ref{cyclotcosetop}, we present on Table~\ref{tab:tableofcycliccodes} all possible scenarios of associated square cyclic codes closed under Schur product.
\exampleend
\end{example}

\begin{sidewaystable}[t]
\caption{All binary cyclic codes $\code{C}_1 = \langle g_1 (x)\rangle$ and the respective generator polynomials of their squares }
  \label{tab:tableofcycliccodes}
\centering
\begin{tabular}{|c|c|c|c|c|c|}
\hline
$I_1$                                       & $g_1 (x)$                               & $I_2$                                      & $g_2 (x)$                               & $I_3$                                      & $g_3 (x)$                               \\ \hline
$\mathbb{C}_0$                              & $p^1 (x) \cdot p^5 (x) \cdot p^{25}(x)$ & $ \mathbb{C}_0$                            & $p^1 (x) \cdot p^5 (x) \cdot p^{25}(x)$ & $ \mathbb{C}_0$                            & $p^1 (x) \cdot p^5 (x) \cdot p^{25}(x)$ \\ 
$\mathbb{C}_0 \cup \mathbb{C}_1$            & $p^5 (x)\cdot p^{25}(x)$                & $\mathbb{Z}_{125}$                         & $1$                                     & $\mathbb{Z}_{125}$                         & $1$                                     \\ 
$\mathbb{Z}_{125}\setminus \mathbb{C}_{25}$ & $p^{25}(x)$                             & $ \mathbb{Z}_{125}$                        & $1$                                     & $ \mathbb{Z}_{125}$                        & $1$                                     \\ 
$ \mathbb{Z}_{125}$                         & $1$                                     & $\mathbb{Z}_{125}$                         & $1$                                     & $ \mathbb{Z}_{125}$                        & $1$                                     \\ 
$\mathbb{C}_0 \cup \mathbb{C}_5 $           & $p^1 (x) \cdot p^{25}(x)$               & $\mathbb{Z}_{125}\setminus \mathbb{C}_{1}$ & $p^1 (x)$                               & $\mathbb{Z}_{125}\setminus \mathbb{C}_{1}$ & $p^1 (x)$                               \\ 
$\mathbb{C}_0 \cup \mathbb{C}_{25} $        & $p^1 (x) \cdot p^{5}(x)$                & $\mathbb{C}_0 \cup \mathbb{C}_{25} $       & $p^1 (x) \cdot p^{5}(x)$                & $\mathbb{C}_0 \cup \mathbb{C}_{25} $       & $p^1 (x) \cdot p^{5}(x)$                \\ 
$\mathbb{Z}_{125}\setminus \mathbb{C}_{1}$  & $p^1 (x)$                               & $\mathbb{Z}_{125}\setminus \mathbb{C}_{1}$ & $p^1 (x)$                               & $\mathbb{Z}_{125}\setminus \mathbb{C}_{1}$ & $p^1 (x)$                               \\ 
$\mathbb{C}_1$                              & $p^0 (x) \cdot p^5 (x) \cdot p^{25}(x)$ & $\mathbb{Z}_{125}$                         & $1$                                     & $\mathbb{Z}_{125}$                         & $1$                                     \\ 
$\mathbb{C}_1 \cup \mathbb{C}_5$            & $p^0 (x) \cdot p^{25}(x)$               & $\mathbb{Z}_{125}$                         & $1$                                     & $\mathbb{Z}_{125}$                         & $1$                                     \\ 
$\mathbb{C}_1 \cup \mathbb{C}_{25}$         & $p^0 (x) \cdot p^{5}(x)$                & $\mathbb{Z}_{125}$                         & $1$                                     & $\mathbb{Z}_{125}$                         & $1$                                     \\ 
$\mathbb{Z}_{125}\setminus \mathbb{C}_{0}$  & $p^0 (x)$                               & $ \mathbb{Z}_{125}$                        & $1$                                     & $ \mathbb{Z}_{125}$                        & $1$                                     \\ 
$\mathbb{C}_5$                              & $p^0 (x) \cdot p^1 (x) \cdot p^{25}(x)$ & $ \mathbb{Z}_{125}\setminus \mathbb{C}_1$  & $p^1 (x)$                               & $ \mathbb{Z}_{125}\setminus \mathbb{C}_1$  & $p^1 (x)$                               \\ 
$\mathbb{C}_5 \cup \mathbb{C}_{25}$         & $p^0 (x) \cdot p^1 (x)$                 & $\mathbb{Z}_{125}\setminus \mathbb{C}_{1}$ & $p^1 (x)$                               & $\mathbb{Z}_{125}\setminus \mathbb{C}_{1}$ & $p^1 (x)$                               \\ 
$\mathbb{C}_{25}$                           & $p^0 (x) \cdot p^1 (x) \cdot p^{5}(x)$  & $\mathbb{C}_0 \cup \mathbb{C}_{25}$        & $p^1 (x) \cdot p^{5}(x)$                & $\mathbb{C}_0 \cup \mathbb{C}_{25}$        & $p^1 (x) \cdot p^{5}(x)$                \\ \hline
\end{tabular}
\end{sidewaystable}

In Example~\ref{Exsqcycliccodes}, each chain of square cyclic codes stabilizes at the second level, i.e., $\code{C}_2 ^2 = \code{C}_{2}$. Still, if $L$ increases and the same parameters are kept, the respective larger chain of square cyclic codes also stabilizes at the second level (see the second row of Table~\ref{tab:tableofcycliccodes}). 
We show that $2$ being a primitive root of unity modulo $p^m$ is necessary for the chain of square cyclic codes to stabilize at the second level.
\begin{proposition}
Let $2$ be a primitive root of unity modulo $p^m$. All chains of cyclic codes closed under the Schur product stabilize at the second level.
\end{proposition}

\begin{proof}
Assume a $L-$length chain of cyclic codes $\code{C}_1 \subset \code{C}_2 \subset ...\subset \code{C}_L$, where the generator polynomial $g_1 (x)$ has defining set $J_1 = \mathbb{Z}_{p^m} \setminus \left\{\{0\} \cup \mathbb{C}_{p^{i_1}} \cup \ldots \cup \mathbb{C}_{p^{i_r}} \right\}$ or $J_1 ^* = \mathbb{Z}_{p^m} \setminus \left\{ \mathbb{C}_{p^{i_1}} \cup \ldots \cup \mathbb{C}_{p^{i_r}} \right\}$, where $i_1 < i_2 < ...< i_r \in \mathbb{Z}_{p^m}$. Then, $I_1 ^*  =\left\{ \mathbb{C}_{p^{i_1}} \cup \ldots \cup \mathbb{C}_{p^{i_r}} \right\}$ or $I_1  =\left\{\{0\}\cup \mathbb{C}_{p^{i_1}} \cup \ldots \cup \mathbb{C}_{p^{i_r}} \right\}$. We prove only to $I_1 ^* $, since the verification for $I_1 $ is the same.

According to Lemma~\ref{cyclotcosetop}, we have the following description to $I_2 = 2 I_1 ^* = I_1 ^* + I_1 ^*$:
{\small
\begin{eqnarray*}
I_2&=& \left(\mathbb{C}_{p^{i_1}} + \mathbb{C}_{p^{i_1}} \right) \cup \left(\mathbb{C}_{p^{i_1}} + \mathbb{C}_{p^{i_2}} \right)\cup \ldots \cup \left(\mathbb{C}_{p^{i_1}} + \mathbb{C}_{p^{i_r}}\right) \cup\\
&\vdots&\\
&\cup& \left(\mathbb{C}_{p^{i_r}} + \mathbb{C}_{p^{i_1}} \right) \cup \left(\mathbb{C}_{p^{i_r}} + \mathbb{C}_{p^{i_2}} \right)\cup \ldots \cup \left(\mathbb{C}_{p^{i_r}} + \mathbb{C}_{p^{i_r}}\right)\\
&=& 2\mathbb{C}_{p^{i_1}} \cup  \mathbb{C}_{p^{i_1}}\cup \ldots \cup \mathbb{C}_{p^{i_1}} \cup \mathbb{C}_{p^{i_2}} \cup  2\mathbb{C}_{p^{i_2}}\cup \ldots \cup \mathbb{C}_{p^{i_2}} \cup\ldots \cup\mathbb{C}_{p^{i_1}} \cup \mathbb{C}_{p^{i_2}} \cup \ldots \cup 2\mathbb{C}_{p^{i_r}}\\
&=& \mathbb{C}_{p^{i_1}} \cup \mathbb{C}_{p^{i_2}} \cup\ldots \cup \mathbb{C}_{p^{i_r}} \cup (2\mathbb{C}_{p^{i_1}} \cup \ldots \cup 2\mathbb{C}_{p^{i_r}})\\
&=& \mathbb{C}_{p^{i_1}} \cup \mathbb{C}_{p^{i_2}} \cup\ldots \cup \mathbb{C}_{p^{i_r}} \cup 2\mathbb{C}_{p^{i_1}}\\
&=& \mathbb{C}_{p^{i_1}} \cup \mathbb{C}_{p^{i_2}} \cup \ldots \cup \mathbb{C}_{p^{i_r}} \cup (\{0\} \cup \mathbb{C}_{p^{i_1}} \cup  \mathbb{C}_{p^{i_1 +1}} \cup \ldots \cup \mathbb{C}_{p^{m-1}})  = 2\mathbb{C}_{p^{i_1}}.
\end{eqnarray*}}

Therefore, $I_3 = 2I_2 = 2\mathbb{C}_{p^{i_1}} + 2\mathbb{C}_{p^{i_1}} = 2\mathbb{C}_{p^{i_1}}$, namely, $I_3 =I_2$. Naturally, one observes that $I_i = I_{i-1}$, for all $i\geq 3$, and the result follows.
\end{proof}

\begin{example}
Let $q=2, L=3$, and $n=7^3 =343$. Note that $2$ is not a primitive root of unit modulo $343$.  Let $\mathbb{C}_1$, where $|\mathbb{C}_1| = 147< \varphi(343)$. Assuming $I_1 = \mathbb{C}_1$, then $I_2 = 2 I_1 = \mathbb{C}_1 \cup \mathbb{C}_3$ and $I_3 = 2I_2 = \mathbb{Z}_{343}$. The chain of square cyclic codes stabilizes at the third level, not at the second. 
\exampleend
\end{example}

\section{Families of $\mathbb{Z}_{2^L}$-additive codes and their $\mathbb{Z}_{2^L}$-linear codes}
\label{sec:families_gray_codes}

We work now with some known families of $\mathbb{Z}_{2^L}$-additive codes and verify whether their $\Integers_{2^L}$-linear codes are linear or nonlinear with results from Secs.~\ref{subsec:linearity-decomposition} and~\ref{subsec:nonlinearity}. Some codes presented here had their (non)linearity verified with different arguments in~\cite{nonlinearity, CordobaVelaVillanueva19}.

\subsection{Hadamard Codes}

\begin{definition}[Hadamard code]
Let $\set{T}_i = \left\{j\cdot 2^{i-1}: 0\leq j \leq 2^{L-i+1} -1\right\}$ for all $i \in [1:L]$. Let $t_1, t_2,..., t_L$ be non-negative integers with $t_1 \geq 1.$ Consider the matrix $\mat{A}^{t_1 , ..., t_L}$ whose columns are vectors of the form $\vect{z}^T,$  where $\vect{z}\in \set{T}_L ^{t_L} \times \dots \times  \set{T}_2 ^{t_2} \times \set{T}_1 ^{t_1 -1} \times \{1\}$. The code generated by $\mat{A}^{t_1 , ..., t_L}$ is the $\mathbb{Z}_{2^L}$-additive Hadamard code $\code{H}^{t_1, t_2 , \ldots, t_L}$.    
\end{definition}

\begin{example}\cite[Ex.~22]{CordobaVelaVillanueva19}
\label{ex:hadamard-generator}
For $L=3$, we have $\set{T}_1 =\{0,1,2,3,4,5,6,7\},\,\set{T}_2 =\{0,2,4,6\}$, and $\set{T}_3=\{0,4\}$. From these sets, it is possible to define the following matrices corresponding to generator matrices of $\mathbb{Z}_8$-additive Hadamard codes:
\begin{IEEEeqnarray*}{c}
   \mat{A}^{1,0,1}=\begin{pNiceMatrix}
    0 & 4 \\
    1 & 1 
\end{pNiceMatrix}, \, \mat{A}^{2,0,0}=\begin{pNiceMatrix}
     0 & 1 & 2 & 3 & 4 & 5 & 6 & 7 \\
     1 & 1 & 1 & 1 & 1 & 1 & 1 & 1 
\end{pNiceMatrix}, \, \mat{A}^{1,1,1}=\begin{pNiceMatrix}
     0 & 0 & 0 & 0 & 4 & 4 & 4 & 4 \\
     0 & 2 & 4 & 6 & 0 & 2 & 4 & 6  \\
     1 & 1 & 1 & 1 & 1 & 1 & 1 & 1 
\end{pNiceMatrix}. 
\end{IEEEeqnarray*}
The decomposition code of the $\mathbb{Z}_{8}$-additive Hadamard code $\code{H}^{1,0,1}$ is
\begin{IEEEeqnarray*}{rCl}
    \code{B} &=& \{(0,0,0,0,0,0), (1,1,0,0,0,0), (0,0,1,1,0,0), (1,1,1,1,0,0), (0,0,0,0,1,1),   \\
    & & (1,1,0,0,1,1), (0,0,1,1,1,1), (1,1,1,1,1,1), (0,0,0,0,0,1), (1,1,0,0,0,1),   \\
    & & (0,0,1,1,0,1), (1,1,1,1,0,1), (0,0,0,0,1,0), (1,1,0,0,1,0), (0,0,1,1,1,0), \\
    & & (1,1,1,1,1,0)\} \subseteq \mathbb{Z}_2^6. 
\end{IEEEeqnarray*}
For all $\vect{c} = \sum_{i=1}^3 2^{i-1}\vect{u}_i,\vect{d} = \sum_{i=1}^3 2^{i-1}\vect{v}_i \in \code{H}^{1,0,1}$, the vectors of the form $(\vect{0},\vect{u}_1 \circ \vect{v}_1, \vect{u}_2 \circ \vect{v}_2)$ define the set $\{(0,0,0,0,0,0),(0,0,0,0,1,1),(0,0,1,1,0,0), (0,0,1,1,1,1)\} \in \code{B}$ and by Theorem~\ref{thm:iff_condition_concatenated} the $\Integers_{8}$-linear code $\gray(\code{H}^{1,0,1})$ is linear.\hfill\exampleend

\end{example}

We recall the recursive construction of Hadamard codes from~\cite{CordobaVelaVillanueva19}. Let $\mat{A}^{1,0,\dots,0}=\langle 1 \rangle$, i.e., $\code{H}^{1,0,\dots,0} = \Integers_{2^L}$. Given $\mat{A}=\mat{A}^{t_1, \dots, t_L}$, we construct for any $i \in [1: L]$,
\begin{IEEEeqnarray}{c}
\label{eq:recursive-hadamard}
\mat{A}_i = \begin{pNiceMatrix}
0 \cdot \vect{2^{i-1}} & 1 \cdot \vect{2^{i-1}} & \dots & (2^{L-i+1}-1) \cdot \vect{2^{i-1}} \\
\mat{A} & \mat{A} & \dots & \mat{A}
\end{pNiceMatrix}.
\end{IEEEeqnarray}
Consider now $\mat{A}_i = \mat{A}^{t_1', \dots, t_L'}$, where $t_j'=t_j$ for $j \neq i$ and $t_i'=t_i+1$. Therefore, the $\Integers_{2^L}$-additive Hadamard code $\code{H}^{t_1, \dots, t_L}$ is the one generated by $\mat{A}^{t_1, \dots, t_L}$.

\begin{example}
\label{ex:hadamard-recursive}
The generator matrix $\mat{A}^{2,0,0}$ of the $\Integers_8$-additive Hadamard code $\code{H}^{2,0,0}$, as previously presented in Example~\ref{ex:hadamard-generator}, can be recursively constructed from $\mat{A}^{1,0,0}$ as
\begin{IEEEeqnarray*}{c}
\mat{A}^{2,0,0} = \begin{pNiceMatrix}
0 & 1 & \dots & 7 \\
\mat{A}^{1,0,0}  & \mat{A}^{1,0,0}  & \dots & \mat{A}^{1,0,0} 
\end{pNiceMatrix} = \begin{pNiceMatrix}
     0 & 1 & 2 & 3 & 4 & 5 & 6 & 7  \\
     1 & 1 & 1 & 1 & 1 & 1 & 1 & 1 
\end{pNiceMatrix}.
\end{IEEEeqnarray*}
\\[-10mm]\hfill\exampleend
\end{example}

Studies on the linearity of $\Integers_{2^L}$-linear codes obtained from $\Integers_{2^L}$-additive Hadamard codes can be found in~\cite{CordobaVelaVillanueva19,FernandezCVelaVillanueva_16,FernandezCVelaVillanueva_20}. It is known that $\gray(\code{H}^{1,0,\dots,0})$ is linear, from~\cite[Ex.~3]{CordobaVelaVillanueva19}.  
We now demonstrate the nonlinearity of a $\Integers_{2^L}$-linear code obtained from the $\Integers_{2^L}$-additive Hadamard code $\code{H}^{k,0,\dots,0}$, for $k \geq 2$ and $L \geq 3$.

\begin{theorem}
\label{thm:nonlinear-gray-hadamard}
The $\Integers_{2^L}$-linear code $\gray(\code{H}^{k,0,\dots,0})$ obtained from the Hadamard code $\code{H}^{k,0,\dots,0}$, with $k \geq 2$ and $L\geq 3$ is nonlinear.
\end{theorem}

\begin{proof} Consider the case $k=2$ and $L=3$, which is the $\Integers_8$-additive Hadamard code $\code{H}^{2,0,0}$ (see Example~\ref{ex:hadamard-recursive}). Notice that the associated code $\code{C}_2$ is given by
\begin{IEEEeqnarray*}{rCl}
\code{C}_2 & = & \{ (0,0,0,0,0,0,0,0), (0,0,1,1,0,0,1,1), (0,1,0,1,0,1,0,1), (0,1,1,0,0,1,1,0),  \\
& & (1,1,1,1,1,1,1,1), (1,1,0,0,1,1,0,0), (1,0,1,0,1,0,1,0, (1,0,0,1,1,0,0,1) \}.    
\end{IEEEeqnarray*}
We can verify that for $\vect{u} = (0,0,1,1,0,0,1,1), \vect{v} = (0,1,0,1,0,1,0,1) \in \code{C}_2$, $4( \vect{u} \circ \vect{v} )\notin \code{H}^{2,0,0}$ and from Proposition~\ref{prop:connections-hammons-etal}, the code $\gray(\code{H}^{2,0,0})$ is nonlinear. Due to the recursive nature of the generator matrix of the Hadamard code, for $L>3$ and $k \geq 2$ the first eight columns of $\mat{A}^{k,0,\dots,0}$ are composed by row blocks of $\vect{0}$ and the last row block is $\mat{A}^{2,0,0}$, meaning that the first eight coordinates are exactly the coordinates of $\code{H}^{2,0,0}$. Since $\gray(\code{H}^{2,0,0})$ is nonlinear, the proof is complete.
\end{proof}

The result of Theorem~\ref{thm:nonlinear-gray-hadamard} will be of particular interest for the nonlinearity of $\Integers_{2^L}$-linear codes obtained from MacDonald codes (more details in Sec.~\ref{sec:macdonald}).

\subsection{Simplex Codes}

Next we present the constructions of $\mathbb{Z}_{2^L}$-additive simplex codes of type $\alpha$ and $\beta$~\cite{nonlinearity}.

\begin{definition}[Simplex code of type $\alpha$]
\label{def:simplex_code_alpha}
The $\mathbb{Z}_{2^L}$-additive simplex code of type $\alpha$, denoted as $\code{S}_k ^{\alpha}$, is the code whose generator matrix is inductively constructed as
\begin{IEEEeqnarray}{c}
\mat{G}_1 ^{\alpha} = \begin{pNiceMatrix} 0 & 1 & 2 & \ldots & 2^L -1 \end{pNiceMatrix}, \text{ and } \mat{G}_k ^{\alpha}=\left(\begin{array}{ccccc}
     \mathbf{0}& \mathbf{1}& \mathbf{2}& \ldots & \mathbf{2^L -1}  \\
     \mat{G}_{k-1}^{\alpha} & \mat{G}_{k-1}^{\alpha} & \mat{G}_{k-1}^{\alpha} & \ldots & \mat{G}_{k-1}^{\alpha}
\end{array}\right),       
\end{IEEEeqnarray}
$L \geq 2$, where $\mathbf{0}, \mathbf{1}, \mathbf{2}, \ldots, \mathbf{2^L -1}$ are the all-zero, all-one, $\ldots$ , all-$(2^L -1)$ vectors, respectively.
\end{definition}

\begin{definition}[Simplex code of type $\beta$]
\label{def:simplex_code_beta}
The $\mathbb{Z}_{2^L}$-additive simplex code of type $\beta$, denoted as $\code{S}_k^{\beta}$, is the code whose generator matrix is inductively constructed as
\begin{IEEEeqnarray}{rCl}
\label{eq:generator-simplexbeta}
\mat{G}_1^{\beta} &=& \begin{pNiceMatrix} 1 \end{pNiceMatrix},\,\, \mat{G}_2 ^{\beta}=\begin{pNiceMatrix}
     \mathbf{1}& 0& 2& \ldots & 2^L -2  \\
     0 \,\,1\,\, 2\,\, \ldots\,\, 2^L -1 & 1 & 1 & \ldots & 1
\end{pNiceMatrix}, \text{ and} \nonumber \\
\mat{G}_k ^{\beta}&=&\begin{pNiceMatrix}
     \mathbf{1}& \mathbf{0} & \mathbf{2} & \ldots & \mathbf{2^L -2}  \\
     \mat{G}_{k-1}^{\alpha} & \mat{G}_{k-1}^{\beta} & \mat{G}_{k-1}^{\beta} & \ldots &\mat{G}_{k-1}^{\beta}  \\
\end{pNiceMatrix},  
\end{IEEEeqnarray}
for $L\geq 2$, where $\mathbf{0}, \mathbf{1}, \mathbf{2}, \ldots, \mathbf{2^L -1}$ are the all-zero, all-one, $\ldots$ , all-$(2^L -1)$ vectors, respectively.
\end{definition}

The Gray image $\gray(\code{S}_k ^{\beta})$ of the simplex codes of type $\beta$ is uniformly packed~\cite[p.~240]{GuptaBhandariLal_05}. That is, given the covering radius  $\rho$ of $\gray(\code{S}_k ^{\beta})$, there exists $\kappa_i \in \Rationals$, $i\in[1:\rho]$, such that for any $\vect{x} \in \Integers_2^n$, $\sum_{i=0}^{\rho} \kappa_i f_i(\vect{x}) =1$, where $f_i(\vect{x})$ is the number of codewords of $\gray(\code{S}_k ^{\beta})$ at distance $i \in [1:\rho]$ from $\vect{x}$~\cite{BassalygoZaitsevZinovev_74}.

\begin{example}\label{excode}
For $L=3$ and $k=1$, let $\mat{G}_1^{\alpha}=(0\,\,1\,\,2\,\,3\,\,4\,\,5\,\,6\,\,7)\in \mathbb{Z}_8 ^8$ be the generator matrix of $\code{S}_1 ^{\alpha}$.
Thus, the chain of binary linear codes $\code{C}_1 \subseteq \code{C}_2 \subseteq \code{C}_3$ associated with $\code{S}_1 ^{\alpha}$ (see Example~\ref{ex:tela-preta-gustavo}) is
\begin{IEEEeqnarray*}{rCl}
\code{C}_1&=&\left\{(0,0,0,0,0,0,0,0),(0,1,0,1,0,1,0,1)\right\},\\
\code{C}_2&=&\left\{(0,0,0,0,0,0,0,0),(0,1,0,1,0,1,0,1),(0,0,1,1,0,0,1,1) ,(0,1,1,0,0,1,1,0)\right\}\\
\code{C}_3&=&\left\{(0,0,0,0,0,0,0,0),(0,1,0,1,0,1,0,1),(0,0,1,1,0,0,1,1) ,(0,1,1,0,0,1,1,0),\right.\\
&&\left.(0,0,0,0,1,1,1,1), (0,0,1,0,1,1,0,1),(0,1,0,1,1,0,1,0),(0,1,1,1,1,0,0,0) \right\}.
\end{IEEEeqnarray*}

Notice that $(0,1,0,1,0,1,0,1),(0,0,1,1,0,0,1,1)\in \code{C}_2$, but 
\begin{IEEEeqnarray*}{c}
(0,1,0,1,0,1,0,1) \circ (0,0,1,1,0,0,1,1)= (0,0,0,1,0,0,0,1)\not\in \code{C}_3, 
\end{IEEEeqnarray*}
namely, $\code{C}_2^2 \not\subset \code{C}_3$. Therefore, from Corollary~\ref{coro:princ}, $\gray\left(\code{S}_1 ^{\alpha}\right)$ is nonlinear. \hfill\exampleend
\end{example}

\begin{example}
\label{ex31}
For $L=k=2$, the $\Integers_{2^L}$-linear code from the $\mathbb{Z}_4$-additive code $\code{S}_2 ^{\alpha}$ is not linear. Indeed, from the generator matrix 
\begin{IEEEeqnarray}{c}
\label{eq:generator-simplex}
\mat{G}_{2} ^{\alpha} = \left(\begin{array}{cccccccccccccccc} 0& 0& 0& 0 &1& 1& 1& 1&2& 2& 2& 2&3& 3& 3& 3 \\ 0&1&2&3&0&1&2&3&0&1&2&3&0&1&2&3\end{array} \right),
\end{IEEEeqnarray}
and consider $\code{C}_1 ,\code{C}_2$ the respective binary codes associated to $\code{S}_2 ^{\alpha}$. We can verify that, given the following codewords of $\code{C}_1$ obtained from the basis (row) vectors in~\eqref{eq:generator-simplex},
\begin{IEEEeqnarray*}{c}
\vect{u}_1= (0,0,0,0,1,1,1,1,0,0,0,0,1,1,1,1), ~ \vect{v}_1 =(0,1,0,1,0,1,0,1,0,1,0,1,0,1,0,1),
\end{IEEEeqnarray*} 
$\vect{u}_1 \circ \vect{v}_1 = (0,0,0,0,0,1,0,1,0,0,0,0,0,1,0,1) \not\in \code{C}_2$. From Corollary~\ref{coro:princ}, $\gray\left(\code{S}_2 ^{\alpha}\right)$ is nonlinear.
\hfill\exampleend
\end{example}

We present now an alternative proof of~\cite[Cor.~3.2]{nonlinearity} and \cite[Cor.~3.3]{nonlinearity} using results from Sec.~\ref{sec:linearity_z2l}. The main idea is that the binary 
codes associated with $\mathbb{Z}_{2^L}$-additive simplex codes of type $\alpha$ and $\beta$, with $L\geq 3$, are not closed under Schur product. 
 
\begin{theorem}[{\cite[Cor.~3.2]{nonlinearity}}~revisited]
\label{th31}
The $\Integers_{2^L}$-linear code $\gray(\code{S}_k ^{\alpha})$ obtained from a $\mathbb{Z}_{2^L}$-additive simplex code $\code{S}_k ^{\alpha}$ is nonlinear for all $L\geq 2$ and $k\geq  1$; except for $L=2$ and $k=1$.
\end{theorem}

\begin{proof}
We start by showing that for $L=2$ and $k=1$, $\gray(\code{S}_1^{\alpha})$ is linear. Indeed, such code is generated by $\mat{G}_1^{\alpha} = \begin{pNiceMatrix} 0 & 1 & 2 & 3 \end{pNiceMatrix}$, and $\code{C}_1^2=\code{C}_1=\{(0,0,0,0),(0,1,0,1)\}$. Notice that $2\code{C}_1^2 \subseteq \code{S}_1^{\alpha}$ and from Proposition~\ref{prop:connections-hammons-etal}, it follows that $\gray(\code{S}_1^{\alpha})$ is linear, for $L=2$.

For $L=k=2$, Example~\ref{ex31} shows that $\gray(\code{S}_2 ^{\alpha})$ is nonlinear, and given the recursive structure of the simplex code~\ref{def:simplex_code_alpha}, for $L=2$ and any $k>2$, $\gray(\code{S}_k ^{\alpha})$ is also nonlinear.

For $L\geq 3$ and $k\geq 1$, let $\mat{G}_k ^{\alpha}$ be the generator matrix of a $\mathbb{Z}_{2^L}$-additive simplex code $\code{S}_k ^{\alpha}$. The first $2^L$ coordinates of each row are composed by zeros; except the last one, which has $0,\,1,\,2,\,\ldots\, ,2^L -1 $ on its first $2^L$ coordinates. Thus, we can see that the first $2^L$ coordinates of any codeword of $\code{S}_k ^{\alpha}$ ary determined by the codewords of $\code{S}_1 ^{\alpha}$ and if the chain of binary codes associated with $\code{S}_1 ^{\alpha}$ is not closed under Schur product, consequently the chain associated with $\code{S}_k^{\alpha}$ is not closed under Schur product either. Based on Corollary~\ref{coro:princ}, we can then conclude that $\gray\left(\code{S}_k ^{\alpha}\right)$ is nonlinear.

It is enough to verify that the chain of binary codes associated to the first eight coordinates of $\code{S}_1 ^{\alpha}$ is not closed under Schur product, for $L\geq 3$. These associated codes are the ones from Example~\ref{excode} and already proven not to be closed under Schur product. Thus, from Corollary~\ref{coro:princ}, the respective $\Integers_{2^L}$-linear code is nonlinear. 
\end{proof}

\begin{theorem}[{\cite[Cor.~3.3]{nonlinearity}}~revisited]
\label{th32}
The $\Integers_{2^L}$-linear code $\gray(\code{S}_k ^{\beta})$ obtained from a  $\mathbb{Z}_{2^L}$-additive simplex $\code{S}_k ^{\beta}$ code is nonlinear for all $L\geq 2$ and $k\geq  2$.    
\end{theorem}
\begin{proof}
%
Consider the case $L=2$ and $k=2$. The generator matrix of $\code{S}_2 ^{\beta}$ is 
\begin{IEEEeqnarray}{c}
\label{eq:generator-k2-l2}
\mat{G}_2^{\beta} =\begin{pNiceMatrix}
   1 & 1 & 1 & 1 & 0 & 2  \\
   0 & 1 & 2 & 3 & 1 & 1  
\end{pNiceMatrix}.   
\end{IEEEeqnarray}
Therefore, the code $\code{S}_2 ^{\beta}$ has the following two associated codes
\begin{eqnarray*}
\code{C}_1 &=&\left\{(0,0,0,0,0,0),(0,1,0,1,1,1),(1,0,1,0,1,1),(1,1,1,1,0,0) \right\},\\
\code{C}_2 &=&\left\{(0,0,0,0,0,0),(0,1,0,1,1,1),(1,0,1,0,1,1),(1,1,1,1,0,0),(0,0,1,1,0,0),\right.\\
&&\left. (0,1,1,0,1,1),(0,0,1,1,1,0),(0,0,0,0,0,1),(0,1,1,0,0,1),(0,1,0,1,1,0),\right.\\
&&\left. (1,0,0,1,1,1),(1,1,0,0,0,0),(1,0,0,1,0,1),(1,0,1,0,1,0),(1,1,0,0,1,0),\right.\\
&&\left.(1,1,1,1,0,1)\right\}.
\end{eqnarray*}
Notice that, for $(0,1,0,1,1,1), (1,0,1,0,1,1)\in \code{C}_1$, we have $(0,1,0,1,1,1)\circ (1,0,1,0,1,1)=(0,0,0,0,1,1)\not\in \code{C}_2$, namely, $\code{C}_1  \subset \code{C}_2$, but $\code{C}_1 ^{2}\not\subset \code{C}_2$. We can then conclude, from Corollary~\ref{coro:princ} that $\gray(\code{S}_2 ^{\beta})$ is a nonlinear code. For $L \geq 2$ and $k \geq 2$, excluding $L=k=2$, observe that in~\eqref{eq:generator-simplexbeta}, there is a submatrix/block $\mat{M}_0^{\beta} = \begin{pNiceMatrix} \mathbf{0} \\ 
\mat{G}^{\beta}_{k-1} \end{pNiceMatrix}$ in the matrix $\mat{G}_k^{\beta}$ generating $\code{S}_k ^{\beta}$. By expanding $\mat{M}_0^{\beta}$, we get
\begin{IEEEeqnarray*}{c}
\label{eq:expansion-M0}
\mat{M}_0^{\beta} = \begin{pNiceMatrix} \mathbf{0} \\ 
\mat{G}^{\beta}_{k-1} \end{pNiceMatrix} =   \begin{pNiceMatrix} \mathbf{0} & \mathbf{0} & \dots & \mathbf{0}  \\ 
\vect{1} & \vect{0} & \dots & \vect{2^{L}-2} \\
\mat{G}^{\alpha}_{k-2} & \mat{G}^{\beta}_{k-2} & \dots & \mat{G}^{\beta}_{k-2} \end{pNiceMatrix} \stackrel{(\dots)}{=}  \begin{pNiceMatrix} \mathbf{0} & \dots & \mathbf{0} & \dots & \mathbf{0}  \\ 
\vect{1} & \dots & \vect{0} & \dots & \vect{2^{L}-2} \\
\vdots & \dots & \vdots & \dots & \vdots \\
\mat{G}^{\alpha}_{2} & \dots & \mat{G}^{\beta}_{2} & \dots & \mat{G}^{\beta}_{2} \end{pNiceMatrix},
\end{IEEEeqnarray*}
and thus, we can see that there are $3 \cdot 2^{L-1}$ coordinates of the code $\code{S}_k ^{\beta}$ that coincide with the codewords of $\code{S}_2^{\beta}$. Recall that its generator matrix is

\begin{IEEEeqnarray*}{c}
\mat{G}_2 ^{\beta}=\begin{pNiceMatrix}
\label{eq:generator-simplex-beta}
     \mathbf{1}& 0& 2& \ldots & 2^L -2  \\
     0 \,\,1\,\, 2\,\, \ldots\,\, 2^L -1 & 1 & 1 & \ldots & 1
\end{pNiceMatrix}
\end{IEEEeqnarray*}
and up to a permutation, $\mat{G}_2 ^{\beta}$ for $L=2$ as in~\eqref{eq:generator-k2-l2} is a submatrix of $\mat{G}_2 ^{\beta}$ for general $L>2$. Since the associated codes are not closed under Schur product for $L=2$ and $k=2$, the same conclusion applies and $\gray(\code{S}_k ^{\beta})$ is nonlinear for all $L\geq 2$ and  $k\geq  2$.
\end{proof}

\subsection{MacDonald Codes}
\label{sec:macdonald}

The constructions of $\mathbb{Z}_{2^L}$-additive simplex codes of types $\alpha$ and $\beta$ imply in $\mathbb{Z}_{2^L}$-additive MacDonald codes of types $\alpha$ and $\beta$, respectively. They were initially proposed in~\cite{ColbournGupta} for the quaternary case, and extended by~\cite{nonlinearity} for any $L>2$.

\begin{definition}[MacDonald codes]
\label{def:macdonald-codes}
MacDonald codes of type $\alpha$ and $\beta$, denoted respectively by $\code{M}_{k,u}^{\alpha}$ and $\code{M}_{k,u}^{\beta}$, are generated by the following matrices over $\Integers_{2^L}$
\begin{IEEEeqnarray}{C}
\mat{G}_{k,u} ^{\alpha} = 
\begin{pNiceMatrix}
     \mathbf{0} & \mathbf{1}& \mathbf{2}& \ldots & \mathbf{2^L -1}  \\
     \mat{G}_{k-1,u}^{\alpha} & \mat{G}_{k-1}^{\alpha} & \mat{G}_{k-1}^{\alpha} & \ldots & \mat{G}_{k-1}^{\alpha}
\end{pNiceMatrix},~
\mat{G}_{k,u} ^{\beta} = 
\begin{pNiceMatrix}
     \mathbf{1}& \mathbf{0}& \mathbf{2}& \ldots & \mathbf{2^L -2}  \\
     \mat{G}_{k-1}^{\alpha} & \mat{G}_{k-1,u}^{\beta} & \mat{G}_{k-1}^{\beta} & \ldots & \mat{G}_{k-1}^{\beta}
\end{pNiceMatrix}, \nonumber \\
\end{IEEEeqnarray}
where $\mat{G}_{k,u} ^{\alpha}$ is obtained from $\mat{G}_k ^{\alpha}$ deleting the columns associated with $\mat{G}_u ^{\alpha}$, for $u \in [1: k-1]$, namely, $\mat{G}_{k,u} ^{\alpha} = \left(G_k ^{\alpha} \,\, \setminus \,\, \frac{\mathbf{0}}{\mat{G}_{u} ^{\alpha}}\right)$, where $(\mat{A}\setminus \mat{B})$ denotes the submatrix of $\mat{A}$ obtained deleting the submatrix $\mat{B}$ from $\mat{A}$. The definition of $\mat{G}_{k,u} ^{\beta}$ is analogous.
\end{definition}

\begin{example} \label{ex:mac-donald-codes} For $u=k-1$, the MacDonald codes $\code{M}_{k,k-1}^{\alpha}$ and $\code{M}_{k,k-1}^{\beta}$ are respectively generated by
\begin{IEEEeqnarray*}{C}
\label{MacDonald}
\mat{G}_{k,k-1} ^{\alpha} = 
\begin{pNiceMatrix}
     \mathbf{1}& \mathbf{2}& \ldots & \mathbf{2^L -1}  \\
    \mat{G}_{k-1}^{\alpha} & \mat{G}_{k-1}^{\alpha} & \ldots & \mat{G}_{k-1}^{\alpha}
\end{pNiceMatrix},~
\mat{G}_{k,k-1} ^{\beta} = 
\begin{pNiceMatrix}
     \mathbf{1}&  \mathbf{2}& \ldots & \mathbf{2^L -2}  \\
     \mat{G}_{k-1}^{\alpha} & ~ ~ ~ ~ ~ \mat{G}_{k-1}^{\beta} & \ldots & \mat{G}_{k-1}^{\beta}
\end{pNiceMatrix}.  ~ ~ ~ ~ ~\exampleend
\end{IEEEeqnarray*} 
\end{example}

Matrices in Example~\ref{ex:mac-donald-codes} can be written in terms of $A^{k,0,\dots,0}$~\eqref{eq:recursive-hadamard}, according to~\cite{nonlinearity}. Observe that $\mat{G}_k^\alpha$ coincides with $\mat{G}_{k,u}^\alpha$, for $u \in [1: k-1]$, after deleting the first $2^{Lu}$ columns of $\mat{G}_k^\alpha$, i.e., the submatrix such that its first $k-u$ rows are all zeros. Analogously, $\mat{G}_k^\beta$ coincides with $\mat{G}_{k,u}^\beta$, for $u \in [1: k-1]$. Furthermore, 
$
\begin{pNiceMatrix}
\vect{1} \\
\mat{G}_{k-1}^{\alpha} 
\end{pNiceMatrix} = 
\hat{\mat{A}}^{k,0,\dots,0},$
where $\hat{\mat{A}}^{k,0,\dots,0}$
refers to the matrix $\mat{A}^{k,0,\dots,0}$ after moving the last row to the top of the matrix. Then, the generator matrices of the $\Integers_{2^L}$-additive MacDonald codes of type $\alpha$ and $\beta$ can be rewritten as~\cite[(10) and (11)]{nonlinearity}, for $k \in \Naturals$ and $u \in [1: k-1]$,
\begin{IEEEeqnarray}{c}
\label{eq:mac-alpha-reduced}
\mat{G}_{k,u} ^{\alpha} = 
\begin{pNiceMatrix}[margin]
     \mathbf{0} & \Block{2-1}{\hat{\mat{A}}^{k,0,\dots,0}} & \mathbf{2}& \ldots & \mathbf{2^L -1}  \\
     \mat{G}_{k-1,u}^{\alpha} &  \hspace*{1cm}  & \mat{G}_{k-1}^{\alpha} & \ldots & \mat{G}_{k-1}^{\alpha}
\end{pNiceMatrix}, \nonumber \\ 
\mat{G}_{k,u} ^{\beta} = 
\begin{pNiceMatrix}[margin]
     \Block{2-1}{\hat{\mat{A}}^{k,0,\dots,0}} & \mathbf{0}& \mathbf{2}& \ldots & \mathbf{2^L -2}  \\
     \hspace*{1cm} & \mat{G}_{k-1,u}^{\beta} & \mat{G}_{k-1}^{\beta} & \ldots & \mat{G}_{k-1}^{\beta}
\end{pNiceMatrix}.
\end{IEEEeqnarray}

We present next an alternative proof of~\cite[Cor.~5.1]{nonlinearity}, which deals with the nonlinearity of $\mathbb{Z}_{2^L}$-linear MacDonald codes of type $\alpha$ and $\beta$. 
\begin{theorem}[{\cite[Cor.~5.1]{nonlinearity}}~revisited]
The $\Integers_{2^L}$-linear codes $\gray(\code{M}_{k,u} ^{\alpha})$ and $\gray(\code{M}_{k,u} ^{\beta})$ obtained from the MacDonald codes $\code{M}_{k,u} ^{\alpha}$ and $\code{M}_{k,u} ^{\beta}$, respectively, are nonlinear for all $L\geq 2$, $k\geq 2$, and $u \in [1: k-1]$.    
\end{theorem}

\begin{proof} For $L=2$ and any $k \geq 2$, we refer to the proof in~\cite[Th.~2, 3]{ColbournGupta}, where the authors demonstrate that if $\vect{c}, \vect{d}$ refer to the first and second rows of the generator matrix of $\code{M}_{k,u}^\alpha$ (resp. $\code{M}_{k,u}^\beta$), then $2 \left((\vect{c} \bmod 2) \circ (\vect{d} \bmod 2)\right) \notin \code{M}_{k,u}^\alpha$ (resp. $\code{M}_{k,u}^\beta$), and from Proposition~\ref{prop:connections-hammons-etal}, the $\Integers_{2^L}$-linear codes $\gray(\code{M}_{k,u}^\alpha)$ and $\gray(\code{M}_{k,u}^\beta)$ are nonlinear.

For $L \geq 3$ and $k \geq 2$, notice that both matrices in~\eqref{eq:mac-alpha-reduced} contain the block $\hat{\mat{A}}^{k,0,\dots,0}$, which coincides with the generator matrix of $\code{H}^{k,0,\dots,0}$, whose $\Integers_{2^L}$-linear code $\gray(\code{H}^{k,0,\dots,0})$ was proven not to be linear in Theorem~\ref{thm:nonlinear-gray-hadamard}. Hence, it is immediate to conclude that also $\gray(\code{M}_{k,u} ^{\alpha})$ and $\gray(\code{M}_{k,u} ^{\beta})$ are nonlinear codes, for $L\geq 2$, $k\geq 2$. 
\end{proof}

\section{Final Remarks}
\label{sec:conclusion}

In this paper, we addressed the linearity of $\Integers_{2^L}$-linear codes by utilizing binary decomposition and associated codes derived from the binary decomposition of an element in $\Integers_{2^L}$. This approach enabled a systematic method for constructing linear $\Integers_{2^L}$-linear codes from Reed-Muller and cyclic codes. We also provided additional properties for the case of $L=2$ and $3$. We explored the relationship between Lee and Hamming distances of a $\Integers_{2^L}$-additive code and its corresponding $\Integers_{2^L}$-linear code. We revisited several results in the literature concerning simplex, Hadamard, and MacDonald codes. 

In terms of future work,  extending our results to generalizations of Carlet's map using Gray isometries between modular rings in~\cite{Gupta22} or to $\Integers_{p^L}$\cite{HengYue_15} figures as a promising direction. Concerning the latter, a recent paper~\cite{Torres-MartinVillanueva22} have explored the decoding of $\Integers_{p^L}$-linear codes, for $p$ prime, for linear and nonlinear $\Integers_{p^L}$-linear codes. This is also an advantage of the nested construction, since it allows efficient multistage decoding~\cite{ForneyTrottChung00}. 

Moreover, constructions of linear $\Integers_{2^L}$-linear codes in Secs.~\ref{sec:linear-gray-codes-rm} and~\ref{sec:linear-gray-cyclic} are related to lattices. It was already pointed out in Sec.~\ref{sec:linear-gray-codes-rm} that the Reed-Muller codes are used in the multilevel Construction of Barnes-Wall lattices~\cite{HuNebe_20}. One can show that the cyclic nested codes presented in Sec.~\ref{sec:linear-gray-cyclic} yield to cyclic (integer) lattices via Construction C. Cyclic lattices have several applications in cryptography~\cite{Micciancio07,FukshanskySun14,ZhengLiuLu23}.

\section*{Acknowledgment}

The authors thank the Editor and the anonymous referees for their helpful comments and valuable suggestions, which improved the quality of the paper. We also acknowledge the support of Estonian Research Council
via grant PRG2531 and the R\&D+i project PID2021-124613OB-I00 funded by MICIU/AEI/10.13039/501100011033 and FEDER, EU. This work was done partially when MFB and \O Y were at Simula UiB.

\bibliographystyle{plain}
\bibliography{defshort1,biblioCstar}

\end{document}